\documentclass[a4paper,12pt,reqno]{amsart}
\usepackage{amssymb, amscd,enumerate,enumitem,color,
mathtools}

\usepackage[british]{babel}
\usepackage[mathscr]{euscript}
\usepackage{verbatim}

\theoremstyle{plain}
\newtheorem{theo}{Theorem}[section]
\newtheorem{lem}[theo]{Lemma}
\newtheorem{prop}[theo]{Proposition}

\newtheorem{cor}[theo]{Corollary}

\newtheorem{lemma}[theo]{Lemma}
\newtheorem{theorem}[theo]{Theorem}

\theoremstyle{definition}

\newtheorem{definition}[theo]{Definition}

\newcommand{\beq}{\begin{equation}}
\newcommand{\eeq}{\end{equation}}
\renewcommand{\a}{\alpha}

\newcommand{\f}{\varphi}
\newcommand{\g}{\gamma}

\renewcommand{\l}{\lambda}
\renewcommand{\o}{\omega}

\renewcommand{\r}{\rho}
\newcommand{\s}{\sigma}
\renewcommand{\t}{\tau}
\newcommand\vp{\varphi}

\newcommand{\bC}{\mathbb{C}}

\newcommand{\bR}{\mathbb{R}}
\newcommand{\bZ}{\mathbb{Z}}

\newcommand{\bP}{\mathbb{P}}

\newcommand{\bN}{\mathbb{N}}

\renewcommand{\gg}{\mathfrak{g}}

\newcommand\SL{\mathrm{SL}}
\newcommand\SO{\mathrm{SO}}

\newcommand{\cD}{\mathscr{D}}

\newcommand{\cF}{\mathscr{F}}
\newcommand{\cG}{\mathcal{G}}

\newcommand{\cN}{\mathscr{N}}

\newcommand{\cQ}{\mathscr{Q}}

\newcommand{\cS}{\mathscr{S}}

\newcommand{\cU}{\mathscr{U}}
\newcommand{\cV}{\mathscr{V}}

\newcommand{\Sim}{\operatorname{Sim}}

\newcommand{\CO}{\operatorname{CO}}

\newcommand{\Gabor}{{\mathcal G}\!\!\operatorname{\it ab}}
\newcommand{\ste}{\operatorname{st}}
\renewcommand{\=}{{:=}}

\newcommand{\p}{\partial}

\renewcommand{\square}{\kern1pt\vbox
{\hrule height 0.6pt\hbox{\vrule width 0.6pt\hskip 3pt
\vbox{\vskip 6pt}\hskip 3pt\vrule width 0.6pt}\hrule height0.6pt}\kern1pt}

\DeclareMathOperator\vol{vol}
\DeclareMathOperator\Id{Id}

\renewcommand\Re{\operatorname{Re}}
\renewcommand\Im{\operatorname{Im}}
\renewcommand\={:=}

\newcommand{\wt}{\widetilde}
\newcommand{\wh}{\widehat}

\newcommand{\bt}{\begin{theo}\ \ }
\newcommand{\et}{\end{theo}}
\newcommand{\bp}{\begin{prop}\ \ }
\newcommand{\ep}{\end{prop}}
\newcommand{\bc}{\begin{cor}\ \ }
\newcommand{\ec}{\end{cor}}
\newcommand{\bl}{\begin{lem}\ \ }
\newcommand{\el}{\end{lem}}
\newcommand{\bd}{\begin{definition}}
\newcommand{\ed}{\end{definition}}

\newcommand{\be}{\begin{equation}}
\newcommand{\ee}{\end{equation}}

\def\<#1,#2>{\langle\,#1,\,#2\,\rangle}

\newcommand{\arr}{\begin{array}{rlll}}
\newcommand{\ea}{\end{array}}
\newcommand{\bea}{\begin{eqnarray}}
\newcommand{\eea}{\end{eqnarray}}
\newcommand{\bean}{\begin{eqnarray*}}
\newcommand{\eean}{\end{eqnarray*}}

\newcommand{\Np}{{\bf N}}
\newcommand{\Sp}{{\bf S}}
\newcommand{\ster}{\operatorname{st}}

\def\sideremark#1{\ifvmode\leavevmode\fi\vadjust{
\vbox to0pt{\hbox to 0pt{\hskip\hsize\hskip1em
\vbox{\hsize3cm\tiny\raggedright\pretolerance10000
\noindent #1\hfill}\hss}\vbox to8pt{\vfil}\vss}}}

\title[The functional architecture of the early vision]{The functional architecture\\ of the early vision and
neurogeometric models}
 \author[D. V. Alekseevsky and A. Spiro]{Dmitri V. Alekseevsky and  Andrea Spiro}
\begin{document}
\begin{abstract} The initial sections of the paper give a concise presentation, specially designed for a mathematically oriented audience, of some of  the most basic facts on the functional architecture of early  vision. Such  information  is usually scattered in a variety of  papers and books, which  are  not easily accessible  by  non-specialists. Our  goal is thus to offer a handy and  short introduction to  this topics, which might be helpful for researchers willing to enter the area of the applications of modern Differential Geometry in studies on the visual systems, baptized {\it neurogeometry} by J. Petitot. We then offer a survey of three of the most important neurogeometric  models: Petitot's contact model  of  the   primary visual cortex, its extension to  A. Sarti, G. Citti and J. Petitot's symplectic model,  and  P. C. Bressloff and J. D. Cowan's spherical model of  hypercolumns. We finally discuss the  main points of the so-called  ``conformal  model''  for hypercolumns (a model that  was briefly presented in [D.\ V.\ Alekseevsky, {\it Conformal model of hypercolumns in {V}1 cortex and the {M}\"{o}bius group}, in ``Geometric science of information'', pp. 65--72,   Springer, 2021] and  given  in  detail in  [D.\ V.\  Alekseevsky and A.\ Spiro, {\it Conformal models for    hypercolumns  in the primary visual cortex V1},  arXiv 2024]), which can be considered as a synthesis of the symplectic and the spherical  models.
\end{abstract}


\maketitle

\tableofcontents

\section*{Introduction}

\label{intro}
\par

The  human visual system is a  highly complicated,   hierarchically organised system, 
   consisting   of several parts,  the   eyes,  the LGN,  the primary  visual cortex  V1,  the cortices V2,  V3, V4    etc., all  of them related one another  with  strong  feedback. 
Some of the most   fundamental  results  on  the  functional structure of the V1 cortex are due  to D. Hubel    and T. Wiesel, who  
 put forward several  crucial ideas, as for instance   the  distinction between  simple and complex  visual neurons, the  discovery of   the
columnar  structure   of the  V1 cortex  and, in particular,   of the   singular  columns (i.e. the  {\it pinwheels}), etc. They also   introduced   the  notions of  
{\it internal  parameters} and  of   {\it hypercolumns}  and  promoted the crucial idea   that  the  V1
cortex  might   be  mathematically represented  as a  fiber  bundle   over the  retina $R$.    This can be taken as one of the main motivations for the recent activity    in applications    
of differential geometry in the   constructions  of mathematical      models for the  early visual system.\par\
In this paper we would like to   offer  a brief   (and, surely not exhaustive) survey of   a diverse   facts, conjectures and problems, concerning  the visual system,   
with a special attention to the aspects that  have been  (or  that might be  in the future)   be analysed   using  differential geometric tools.  After such preliminary review, we give a short introduction to   {\it neurogeometry},    a  recent  
research area in  applied mathematics,   initiated by  J. Petitot and his coworkers (see e.g. \cite{P-T, P, S-C-P}), that is meant to produce 
   {\it continuous}   models for 
 different  brain  subsystems -- in particular  to the visual system -- using   method and results from  Differential Geometry, Lie
 Groups, Differential  Equations  and  Statistics.\par
\smallskip
In  more detail, in the  first  sections of this paper we  begin recalling   some  basic   facts about 
the functional architecture of the early vision  in both the static setting (that is, with  still eyes and still stimuli) and in dynamics.
We also give a short review of  some  common mathematical representation   of the (linear)oculomotor
visual neurons as filters, providing      differential geometric interpretations. We then address a few  mathematical aspects of
the interactions, occurring  in the process of the visual perception,  between  the oculomotor information  on the eye position and  the information coming from 
the photoreceptors  in the  retina.  In particular,  we  use  Donders' and Listing's laws in order to analyse    the  configuration  space of  the  eye   and
 consider  some of the most    basic  properties   of   the  saccades  and   fixational eye movements,  which  lead to the 
 perception  of  stable   objects. 
 We then discuss  the phenomenon of   the pre-saccadic  shifts of receptive  fields  and  of the   remapping process together with  the problem of identification between two different
   retinal images  after  a remapping.   On this regard, we present a conjecture according to which  such an identification can be  mathematically represented in terms of an appropriate  conformal transformation of the  eye
 sphere and  we show  that such a conjecture  is   a realisation of  the    Etcetera Principle of E. Gombrich. Consequences of all  this  concerning the 
 visual  stability problem are  also  discussed.\par
\smallskip
As we mentioned above, in the last    two sections, our discussion focuses on    some  important differential geometric  models for  the  primary  visual cortex V1 and  for  the
hypercolumns, namely: (a)   the  pioneer   model by  W. Hoffman,
(b)  the contact model introduced by Petitot and Tondut, (c) its  symplectic extension  by Sarti, Citti and Petitot  and  (d)     Bressloff and
Cowan's   spherical  model   of a  hypercolumn.\par
  All these  models    are basically founded on a  general idea  by Hubel and Wiesel, namely 
 that the  firing of a  simple  neuron    with  a receptive  field represented by   a point  $z$ of the retina  $R$,  depends  not just on the  {\it value}   at $z$ of the   input density of energy    of   the incident light,   but   also  
  on several other data of the  local properties    of  such  input  function
  near $z$ (as, for instance,   the {\it orientation} and the  {\it  spatial frequency}).
 The first  who  transposed     this idea  into   differential geometric  terms     were W. C. Hoffman and J. Petitot,   who  mathematically represented
     the  early visual system     in terms of  a     fiber bundle
     $\pi: P \to R$
       over  the  retina $R$.  In such representations, 
    each    point $z$ of the retina $R$  (which is ideally considered  as  a  surface) is   a mathematical representation for  the receptive field (RF) of a column. Indeed, we recall that 
   the RF of a columns is the union of the RFs of the  neurons which  are  contained in it and  that, even if it   is slightly larger than each single RF of  its  neurons, in  first approximation
   it can be considered as a  single point of the retina.  At the same time, the points of each    fiber $\pi^{-1}(z) \subset P$ of a RF $z$  are  mathematical representations of the visual  neurons that are in   the  column with  such RF.   in  the mathematical models of kind,  only simple   neurons are considered.  This   is  motivated by the fact that the first information   coming  from the  retina is  mostly transmitted just  to   simple neurons. \par
   According to this kind of differential geometric models  the coordinates 
 $\theta_1, \cdots. \theta_k$  of the  fibers of a principal bundle   correspond to Hubel and Wiesel's
 ``{\it internal parameters}''  and are  related    with the local   properties of   the   input  energy  function at their  base points $z \in R$, 
      as,  for instance, the orientation,  the  modulus of the gradient, the coordinate in the  color space, etc. As we already mentioned, 
a fixed   simple neuron  $n_{\theta^o_1, \cdots, \theta^o_k}$ of a column with RF $z$  is   mathematically represented  by  the    point  $(\theta^o_1, \cdots, \theta^o_k)$ 
    of a fiber $\pi^{-1}(z)$ and it  is assumed to  fire  only if   the
  restriction of the  input  function  to  its  RF  $z$  has  the local   properties, corresponding to  the   internal parameters $(\theta^o_1, \cdots,
  \theta^o_k)$. \par
  A complete  list  of  the  internal  parameters, which are relevant for the visual system -- and thus the exact number of  fiber coordinates  for the most appropriate  differential geometric model of the form $\pi: P \to R$ for the V1 cortex -- 
is not known. The    most  important internal parameters are surely the  {\it orientation} and the {\it spatial
frequency}, but there are  many other relevant parameters, as for instance those for  the color space,
 the contrast, the curvature, the temporal frequency, the ocular dominance, the disparity and the
direction of the motion.    N. V. 
 Swindale \cite{S}  estimated the  total number of  internal parameters  as  $6-7$  or  $9-10$.\par

Another   important  and deep idea  by Hubel and Wiesel  is that {\it the columns are locally grouped into hypercolumns}, i.e. sets of columns  characterised as follows: 
{\it A hypercolumn $H$ is  a  minimal collection  of   columns  of the V1  cortex  with 
 different reactions  for each of the  possible   values for   the     internal parameters}.  The basic idea   is  that a hypercolumn, or, more
   generally,   a   system of horizontally connected   hypercolumns,  is a collection of cells  that is  responsible  for   the  full perception     of the
   local  structure of a retina image.  On this regard, we have to  mention  that,  contrary to  the columns, the existence of hypercolumns has been debated
for a long time, see \cite{T-Z-B}.
\par
The  first  differential-geometric models of the V1 cortex, proposed  in the papers by  Hoffman,  Petitot   and Tondut, and Petitot,  can be  briefly described as follows. 
According to Hubel and Wiesel's results, the simple neurons of the V1 cortex detect the 
{\it contours}, i.e  the  level sets  of the input  energy function I on the retina R
(that is,  the energy density  of the light that hits the retina) with large gradient. Starting from  this idea, 
Hoffman  was  the first who presented a pioneering model   (with a few mathematical inaccuracies) of the
Visual Cortex in terms of  a contact bundle [20]. 
Inspired by   Hubel,   Wiesel  and Hoffmann's ideas and motivated by the new  experimental discoveries concerning  the structures of the orientation maps
(results that were made possible by  revolutionary Bonh\"offer and Grinvald's techniques of the early
nineties),   J. Petitot and his collaborator Y. Tondut  \cite{P, P1, P-T}  developed  in great detail the so-called  {\it contact model} for the V1 cortex.  
In  such a  model,   the retina $R$,  is  identified with the Euclidean
plane $R = \bR^2$ and the V1 cortex is   represented  as  the projectivised
tangent bundle $\pi : P T R\to R$ of $R = \bR^2$. This bundle admits a natural system of coordinates $(x, y, \theta)$, in which
$(x, y)$
are coordinates for the points of the retina $ R $   and $\theta  \in  [-\pi/2, \pi/2]$
 is the  {\it orientation} of the lines passing through $(x,y)$, i.e. the angle made by  the line  and the
axis $0x$. We remark that $ P T R$ can be also interpreted as the space of the infinitesimal
curves (or, more precisely, the 1-jets of the non-parameterised curves) in $ R = \bR^2$ and  it is
equipped with the canonical contact 1-form  $\eta =dy  - \tan{\theta dx}$.\par
     Petitot's  contact  model 
     was later  combined with
       G. Citti and A. Sarti's  model of  perceptual completions of images  \cite{C-S0}
 to        determine an extended differential-geometric  model of the V1 cortex, which we call  {\it Sarti, Citti  and Petitot's symplectic model}. In this new model the simple cells of the V1 cortex are described in terms  of the points of  a bundle with  two-dimensional  fibers. 
 The two coordinates $(\theta, \s)$ of the fibers (= the internal parameters   considered in this model)  are  the previously defined orientation  $\theta \in  [-\frac{\pi}{2}, \frac{\pi}{2})$  and a  new parameter, the    so-called {\it scaling  factor}\ $\sigma$, which  describes  the intensity   of    response of a neuron to a  stimulus. 
      In \cite{S-C-P}  the  authors propose an interpretation of such scaling factor $\s$   in terms of the distance between the    RF  of  a neuron  (which is assumed to be activated through the  so-called {\it maximal selectivity process})   and the regular  boundary of a  retinal figure.  
      Note    that such interpretation of $\s$  has a non-local character and does not allow  to  consider   it    as an internal parameters in the  sense of Hubel and Wiesel. \par
Aiming to determine an alternative (and purely local) interpretation of the scaling factor,     in \cite{A} the first author proposed the {\it conformal spherical model for hypercolumns}, which is  based on  ideas of P. Bressloff and J. Cowan's theory of hypercolumns  and leads  to an interpretation  of $\s$  in terms of the   {\it normalised   spatial frequency},    another     internal parameter that is    equally fundamental as    the   orientation.  The conformal spherical model   has been later developed in full detail in \cite{A-S}. In that paper, it   is shown    that,  in   small neighbourhoods   of  pinwheels,   the  conformal spherical   model   reduces    to  a reduced
model,   which is  mathematically  identical to     Sarti, Citti and Petitot’s  symplectic  model.  In the last  section of this paper,  we discuss some  important features of  such  a reduced   model and  its  relation  with  the  symplectic model.
\par
 This paper ends with  a short    discussion of     differences between the simple and  the complex neurons.  It is known that each complex neuron of the V1 cortex   collects information  from  systems  of  several simple neurons.  It is also known that    the  simple  neurons  are sensible   to  the  shifts of the  contours    in  their receptive fields,  while  the  complex  neurons  are  not (see  e.g. \cite{H,H-H,Car}).   In our  last section, we   state the  {\it Principle of Invariance} and discuss a possible use of this principle to explain such a fundamental difference.
    \par
   \smallskip    
The  paper is organised as  follows.
In \S \ref{sec:1}, we  offer an outline  of  some  of  the  most  important known  facts on  the    visual system   and 
 the  functional  architecture  of  the primary  visual
cortex V1. Many topics very briefly mentioned in this preliminary section  will be  discussed in greater detail in the subsequent sections.
Models of  the   visual neurons   as filters and some of  their  geometrical interpretation are  discussed in
 \S \ref{gaborsect}.  In \S \ref{sect:3}   we provide a   description of  the  functional  architecture  of  the primary  visual
cortex V1, based on the  fundamental ideas  of Hubel and  Wiesel.
Section \S \ref{sect 4}  is devoted  to   the  geometry of  the  eye movements  (i.e. to the {\it fixational eye movements} and  the {\it saccades}) and to a discussion  of
Donders' and Listing's laws. 
A short introduction to  conformal geometry of the  sphere   is  given in  \S  \ref{section 5}.
 The information processing in  dynamics  is the topic of \S  \ref{sect:6}, where   we  discuss  the  shift of  receptive  fields, the 
 remapping  phenomenon  and  the problem of identification of  retina images before  and  after  the remapping.
   In that section, we state the conjecture that the remapping is determined by a conformal transformation.  
   We then show that  such a conjecture gives   a  realisation  of   the
  Etcetera Principle by E. Gombrich and we discuss   its   consequences for    the Alhazen Visual Stability Problem. 
In  \S \ref{section 7} we provide  a short  exposition of the   contact   models of  the V1 cortex   by Hoffman,   
Petitot  and Tondut, and Petitot of  the   symplectic model of  Sarti, Citti and Petitot.   The concluding section \S \ref{section 8}   is devoted to models of   hypercolumns. 
More precisely,  we  describe  the geometric structure of  Bressloff and Cowan's spherical model,  our    modification of such a  model in  the framework of   conformal geometry  and   a  reduced version    of  such modified model  for    neighbourhoods  of pinwheels.  We then  discuss 
  the  relations         between     Sarti,
Citti and Petitot’s symplectic model,  Bressloff and Cowan’s spherical model and our   conformal model.
 In the  concluding subsection \S \ref{subsection
8.4.}, we state the {\it Principle of Invariance}
and  use  it    for  an explanation      of  the  differences  between simple  and   complex cells.
\par
\medskip
\noindent{\it Acknowledgements.}  We  warmly  thank  J. Petitot  for  valuable comments  and
useful advises and suggestions that helped us to  improve the  overall presentation.

   \section{General  principles  of  the organisation of  the visual  system}
  \label{sec:1}
  \par
In this section,  we list the main   assumptions that we  adopt throughout  the paper and  we  outline    some  well  known facts on  the visual  system, with a particular attention to  geometric  aspects.  Many of the topics, that  are just briefly  touched here,  will  be further discussed  in the next sections.  As for any  survey  of a   very wide area, we are  aware that our exposition has little   chance   to be   fully comprehensive.  
But we  need to  stress  that, as we mentioned above,    our   goal is essentially   to provide  a concise and handy  overview    of the  most 
 important       facts  on the   functional architecture of  early  vision, which can be  helpful for   mathematician, who are interested in neurogeometry and are  not a specialists  in the physiology of vision.   For additional information  on the topics outlined in the following,  the reader is  referred to the   textbooks, reviews and articles on the visual  system  \cite{H,B-T-T,Ko,J-G-H-W,K-B,K,Marr,P, P1, C-S} and
the many references therein.\\[5 pt]
 1.  Our  discussion  is limited  to the   {\it monocular  vision}  given  by gray levels  (no  color).\\[5pt]
 2. In the first part  of this paper,  we are going to  limit the discussion to  the  {\it static} situation  (i.e. when the  eye  and the  stimuli  are    still).
 In this setting,   it can be assumed that  the brain extracts all visual information just from the  {\it retinal input   function $I_R$},   i.e. the density of the energy of the
light which is incident to the retina $R$.
  Such  a  function   is encoded into the excitation of
the photoreceptors (rods and cones).    
Note however that in the  natural setting,   the  eyes  are never still,  they  are 
  continuously  moving    and  it is experimentally proved that,  if compensations of   the 
  eye movements are made, then  a  loss of vision occurs  within $2$-$3$ seconds \cite{Y}.
In the {\it dynamics} setting,  the input energy  function $I_R$  is a function  not only of the points $z \in R$ of the retina 
but also  of the time $t$,  and the vision is the result of an interaction between the stochastic
information about the dynamics of the retinal image, encoded in the photoreceptors,
and the information about the eye position, given by the oculomotor control system.
  Copies of the oculomotor commands on the eyes movements ({\it efferent copies}
or
{\it corollary discharges}) meet the  information coming  from the retina in some region of the cerebral cortex, which is probably the {\it medial superior
temporal (MST)} area  \cite{C-M-W}.
\\[5pt]
3.  We   assume that    the head is  fixed   and we  consider  the   eye ball    as  a rigid  body, which may  rotate
around  its center $O$. The   retina $R$  is considered as   a very large  domain of  the eye  sphere  ( = the  boundary of the
eye  ball). Sometimes  it is convenient  to identify  the  retina  with  the  whole eye sphere.  We  also assume  that the
optical center of  the  eye, i.e.  the {\it nodal point    $\cN$},    is located  on  the  eye  sphere. In  reality, $\cN$
is an  inner  point of the eye  ball,  but  it is very  close  to the boundary.\\[5pt]
4. The  retinal image  of an  external  surface $\cS \subset  \bR^3$ is  obtained by    the { \it    central projection
with respect
to $\cN$} of  the surface $\cS$   onto the eye sphere. Such a projection  is the map $\varphi_{\cN}: \cS \to R$,   which  sends each
point  $A$ of  $\cS$   into the   point $\overline{A}$ of  $R$,  which is determined  by   the  intersection between the   ray
  $\ell_{A\cN}$,  originating from   $A$ and passing through $\cN$,  and  the
retina $R$:
 \begin{equation} \label{central}
   \cS \ni  A \overset \f  \longmapsto 
    \overline{ A}  \=
 \ell_{A\cN}\cap R   .
 \end{equation}
  Under the assumption that each  $A \in \cS$ is the source of an ideal  diffused reflected light,  the 
  corresponding  $ \overline{A}$ is the point in the retina that  receives the light emitted by  $ A$ with   an
  intensity  that depends  on the energy density  at the emission point.\\[5pt]
 5.  We assume that the  {\it center of  the  fovea} is   a  point $\mathcal F$  of  the  eye  sphere,    opposite
 to $\cN$. \\[5pt]
 6.   We  denote by $S^2_{\text{eye}}$ the eye sphere  in   primary position and by $(x,y)$ the Euclidean   coordinates
 of  the  tangent  plane $T_{\mathcal F} S^2_{\text{eye}} $  of  $S^2_{\text{eye}}$  at the center  of the fovea $\cF$.
  The   stereographic  projection with respect to the pole
   ${\cN }$
 $$ \ste_\cN:  S^2_{\text{eye}} \longrightarrow T_\cF S^2_{\text{eye}}  = \bR^2, \qquad    p \longmapsto  \ste_\cN(p) :=
 \ell_{\cN p}\cap T_\cF S^2_{\text{eye}} $$
  allows to consider $(x,y)$  as {\it standard  conformal coordinates}  for the retina $R \subset S^2_{\text{eye}}$.  In
  these coordinates,   the    metric of  the eye sphere $ S^2_{\text{eye}} $,  which is given   by its  embedding  into the physical Euclidean
  space  $E^3 = \bR^3$,
    has  the  form
    $$g =  f(x,y) (dx^2 +  dy^2)$$
      for some function  $f(x,y)$. In  a  small neighbourhood of   $\cF$, the   function $f(x, y)$ is
      approximately constant  and   equal to $1$,
    and   $(x,y)$  are   approximately Euclidean  coordinates  for  $g$. Under   these  approximations,
    the  retina  $R \subset S^2_{\text{eye}}$  can be (locally) identified with the Euclidean plane   $\bR^2$.
    \\[5pt]
8.  The visual   system    has a  hierarchical structure   with  a strong   feedback.
    The  initial  input   function  $I_R$,   which is recorded  by the photoreceptors (cones  and  rods),  is   very irregular.
    The  purpose of the   information  processing in the   retina   is to   regularise and contourise  the
    input  function $I_R$   and  prepare it for  decoding. The output of this process is  a regularised   function  $ I: R
    \simeq \bR^2 \to \bR$,   which  is   the {\it  input (energy) function} (for the visual cortex V1).  It  is  encoded  in  terms of   the excitations of the 
    ganglion cells whose   long  axons     terminate  in   the  {\it LGN}   ({\it Lateral Geniculate  Nucleus}).
      The visual information    is  first sent   from  the retina    to  the LGN  through  the  axons,  second it  is sent  to the
      {\it primary visual  cortex V1}, and  then   to  the  regions  {\it V2}, {\it V3}, etc.,  i.e.  to all other regions of the
      visual  system.   The  transformations which correspond to the changes  from   the input  functions $I_R$  on the retina   into  the
      corresponding  functions $I$ that are the input for   the
                 LGN and  the  V1 cortex, 
               can be   mathematically  described  (in first approximation) as    {\it conformal maps}  with respect to appropriate natural metrics, see \cite{Sch}.
                These transformations are  called    {\it retinotopic}  (or {\it topographic})  {\it mappings}.\\[5pt]
  9.
  The visual information of  the V1 cortex  is encoded  in  the firings  of  the {\it visual neurons}.  Each visual neuron  works   as  a  filter, i.e. as a functional on the  space  of  all possible  input  functions. In static,
    the   firing  of  a visual  neuron depends only on the restriction $I\vert_D$ of the input function $I: R \to \bR$ to
    some small domain $D \subset R$, which is called  the {\it receptive  field  (RF)}  of  the neuron. In dynamics,  the situation is much more complicate.  We  will discuss it in 
\S \ref{sect:6}.\\[5pt]
 10.   {\it  Simple and  complex neurons.} Hubel and  Wiesel   divide   the   neurons    into {\it
 simple}  and  {\it complex}.  A {\it simple neuron}  (also called  {\it simple cell})    works  as a linear  filter, i.e.
 as  a  linear  functional  on the  space of the  input  functions, whose output   is   the
 average    of  the    restriction  $I_D$ of the input  function $I$ to the receptive field   $D$ of the neuron,  in which  the mean value is computed weighting each point  $z \in D$   by the  value $W(z)$  of  an appropriate weight function $W: D \to \bR$, called   {\it receptive  profile  (RP)} of
 the  neuron.  The  receptive profiles of  the simple neurons are  well approximated by   the so-called  {\it Gabor  functions},  i.e.  Gauss functions  modulated
 by  a  sine or a  cosine function (see \S\ref{gaborsect} for details).\par
    The other neurons  are  the     {\it  complex neurons} (or {\it complex cells}). They collect
    information   from   systems  of  several  visual neurons and usually    work as   non-linear functionals.
Several  models for the  complex cells have been   proposed --    one of the  first   models was proposed by   J.\ A.\ Movshon,
I.\ D.\ Thompson and D.\ J.\ Tolhurst \cite{M-T-T} (see also \cite{Car,Car1}). Following an idea by Hubel and Wiesel,
they  describe  a complex cell as a  non-linear filter  for the   information, which is determined  by a system of
(linear)  simple cells with  identical  orientation  but  different receptive  profiles.
    According to this model,  the  processing  performed  by a  complex cell consists of  a
    rectification and a combination  of   the information,  which is collected by  the  simple cells to which it is   connected. A  crucial
    difference between the complex cells and the  simple cells is given by the  fact that  the  reaction of a
    complex neuron to a contour  is invariant under  shifts of that contour  within the receptive field of the neuron.  Till now   this  model  is in a good
    correspondence  with experiments (see \cite{Car}).\par
The  main  limitation of  these  models of complex and simple cells   is that they are static, that is they do not take  into account  the time
evolution of the incoming information.
In dynamics the  situation  becomes much more complicated  and  the  non-linear  neurons (filters)    play  a more
important  role. In this paper  we   are mostly concerned with   the (linear) simple  neurons.  Complex neurons will be  briefly discussed  only     in the last
section.\\[5pt]
11. {\it Processing  of the visual information in the  retina}.
  Before  going     from the  retina  to  the  visual cortex,  a processing of the visual  information  occurs  {\it within the   retina}.  We  do not    discuss in  detail   the    very complicated    system of processing  information in  retina -- for this topic with  refer   to the  excellent   survey \cite{P-B-B}. Here we just mention that, roughly speaking,  the retinal  processing of  the   visual  retinal information, encoded in the   photoreceptors (codes and  rodes),  consists in   regularisations and  contourisations of   the  stimuli.   The first  physiologist who detected    a response on a retinal neuron  and  proposed   a model  for   such actions   was S.  W. Kufler. We will discuss  his  model  in \S \ref{Kuffler}.\\[5pt]
12.  As we mentioned above, after the initial retinal processing, most  of the   retinal information    goes     to    the    LGN (the  {\it Lateral Geniculate  Nucleus}), which is  a  part of the Thalamus.  The LGN does not  perform a further processing of  the information coming  from the retina, but it  supplies   additional instructions  or changes for  some of the   parameters, which are needed   for subsequent manipulations,  and  distributes the updated information  among  several other    subsystems  for    processing. Most of the information  is sent to simple cells of the primary cortex V1. But some  is sent directly   to the higher visual  subsystems V2 and V3.
Another very  important purpose    of the LGN is   to provide  feed-back from the higher visual systems to the primary information processing system, primarily to  the V1 cortex. 
%
%
%
%
           \\[5pt]
 13. 
 There  are three   pathways   from the  retina  to the  primary visual cortex  V1   through the LGN: The {\it P-pathway}, which   is mostly
 responsible   for  the perception  of   stable  objects,   the  {\it M-pathway}, which  is important  for  the perception
 of  moving objects,   and the {\it K-pathway}.    The  function of  the K-pathway is not know properly, but it is  known that it is related  with the perception of colors. 
 In this paper we    consider  only the  {\it P}-pathway.
    The  structure of  the other pathways    is more involved: For instance the {\it M}-neurons (which are responsible for   the {\it
    M}-pathway)  are  non-linear,  even those on  the retina \cite{K-B}.\\[5pt]
14. The energy   function  $I: R \to \bR$ on  the retina  $R$ is   completely determined  by its values  and   the
  level  sets $  L_c= \{z \in R\ : \ I(z) = c \}$ of the assumed values.
     {\it Only the level sets, not the     values,   are    truly  relevant}  (just think about the fact that
    when the light in the room is turned on, the illumination of the retina changes dozens of times).
     This fact is  consistent   with  Hubel  and  Wiesel's discovery  that   {\it the main  objects  detected  in  the
     early  vision are the {\rm  contours}, that is  the (non-parametrised) level  curves $L_c = \{I = c\}$  of  the  energy function  $I$
     with large gradients}.   Note  that, mathematically,  a contour  is an   integral curve
       of  the 1-dimensional distribution,   given by the   Pfaffian system $    dI = 0$,   and it  depends just  on the
       conformal  class $[\omega]$ of the 1-form $\omega = dI$. \\[5pt]
 15.  A contour $L$ though   a point of the retina $z \in R$  is locally  approximated  by     its   tangent  line $ \ell=
 T_z L  \subset T_zR$ at  $z=  (x,y)$ or,   more precisely,  by a small interval $b \subset \ell$ of this line, called {\it bar}.
 The  line $\ell$  and  the bar $b$   are determined  by    their {\it  orientation}  $\theta \in [-\pi/2, \pi/2)$, i.e.
 the  angle  formed by  $ \ell$ and   the     $x$-axis.
    The   space  of  the  {\it infinitesimal  contours} (also called {\it orientations}) is  the space   of  all  tangent lines to
    non-parametrised curves of  $R$, i.e.   the
              projectivised   tangent   bundle  $PTR$ of $R \simeq \bR^2$.  We recall that  an open dense subset of $PTR$ can be   locally identified  with
              the  space of  $1$-jets
              $J^1(\bR, \bR) = \big\{(x,y, \frac{dy}{dx})\big\}$  of  the functions $y = y(x)$ of the  real line.\\[5pt]
  16. The  {\it orientation} $\theta $ of the tangent line of a contour is  a very   important (but  not unique) internal parameter for  the local
  structure of  the input energy function $I$, the main stimulus in  the  early vision. Another fundamental internal parameter is  the   {\it spatial  frequency} 
  \cite{C-R,S-N-R, D}:  Roughly  speaking, it  is a measure of how often the  sinusoidal components of the
  stimulus (given by its  Fourier transform) repeat per unit of length.
   It is  measured by   number of  cycles  per degree,  $c/\operatorname{deg}$ (\footnote{ The  {\it degree}   is the
the standard  length measure for  the  eye sphere.}).  It is  not an infinitesimal  characteristic    of  a  contour $L$,
   but
   of the   structure of  the image in a neighbourhood  of the  contour.  More precisely,  {\it it is a datum which is
   characteristic of  the   $1$-dimensional distribution $\cD = \ker d I $ near  a  contour}. The Fourier  analysis
   allows to  approximate   the  distribution $\cD = \ker d I $  by means of  a {\it sinusoidal  grating}. Such a   grating is
   determined  by  $4$ parameters:   {\it spatial  frequency}, {\it contrast}, {\it orientation}  and {\it phase}.
   All  these parameters   may be  considered as internal  parameters and they    give  information on the local properties 
 of  the image. In this paper,  following  Bressloff and Cowan,  we   focus  on   the  orientation
   and the spatial  frequency.  On this topic, we would like to mention that J. G. Robson   et al.  \cite{S-N-R} showed   that  there  are many  independent channels
   for the spatial frequency, but   only a few of them correlates. For example,   the channel with spatial frequency  $p =
   14\, c/\operatorname{deg}$   correlates     with a  channel  with a  frequency $p'$ if    $p/p'$ is  $4/5$ or $5/4$.
   \\[5 pt]
      17. {\it Hyper-specialisation   of  visual neurons.}
A  visual  neuron   fires  only  when, in its receptive field, 
    the local  internal parameters  of the stimulus (as e.g.  its orientation, spatial frequency. contrast,  etc.)   take (up to   small variations)
      some   prescribed  values,   which are uniquely associated with  the neuron.
      Actually, the   firing  of many  types of  visual neurons  depend  not just  on  such  prescribed values, but  also on  their rate of change. This is the reason why, in order to  perceive a stable
      object,   the eye must constantly move.
      \par
      \medskip
     \section{Models of linear  visual neurons } \label{gaborsect}
 A fundamental  mathematical model for the visual neurons is provided by the following notion.  By {\it  linear neuron with receptive
 field (RF) $D \subset R \simeq \bR^2$}   we mean   a  neuron, which  works  as a linear filter $T_W$ on the input energy function $I(z)$
 of  the  form
 $$I  \overset{T_W}\longmapsto \int_D W(z)I(z) \vol\ ,\qquad   z=(x,y)\ ,$$
 where we denote  by  $\vol :=  dx\,dy$ and  by $W(z)$   a  {\it weight function}  (in  neurophysiology, it   is called  {\it receptive
 profile (RP)})  which  characterises the filter $T_W$.
In mathematical language, $T_W$  is a linear   functional,  determined  by    the  weight  function $W$ with support $D$.
One  may say that    it    calculates a  sort  of  ``mean value'' of the restriction
$I\vert_D$ of the input function to the receptive field $D \subset R$,  where  each   point   $z$  is counted
with  the  weight    $W(z)$.
  \par
     We recall  that  a transformation
      $z \to z^\prime= \varphi(z)$ with positive Jacobian  $J(\varphi) =  \det\big\vert\frac{\partial z^\prime}{\partial
      z}\big\vert
      $,
      changes  the   coordinate   system $z=(x,y)$  into  the new  coordinate  system
      $$ z'(z)=(x', y') = (x \circ \varphi^{-1}, y \circ \varphi^{-1}) \ .$$
  Therefore, under such transformation, each  receptive profile $W(z)$ changes into
 \begin{equation} \label{density}    \varphi(W)(z') :=  \big(J(\vp)\vert_{\vp^{-1}(z')}\big)^{-1}\,W(\vp^{-1}(z'))  \
 ,\end{equation}
 i.e.  {\it the RPs  transform as   densities}.\par
 In the following subsections, we discuss several types of  linear filters that are particular relevant for vision.
\par
\medskip
\subsection{Gauss filters on  the plane}
The {\it Gauss    filters}  are   the linear filters with RP   given by the  Gauss  probability   distributions of  $\bR^2$,
that is  the functions defined as follows   (see also   \S \ref{sect3.5}).
 Let us call   {\it mother Gauss    filter} the linear filter  $T_{\g_0}$
 with  RP
 $$    \g_0(z) :=  \frac{1}{\sqrt{2\pi}} e^{- \frac{1}{2} \vert z\vert^2}\ .$$
A   {\it Gauss  filter}   is the  filter  with the RP that is  obtained from the  $\g_0$  by means of   a
 transformation of  the   oriented  affine group   $\operatorname{Aff}^+(\bR^2)   =
 \operatorname{GL}^+(n) \cdot \bR^n$, i.e. by a transformation 
  $$ z \overset{T_{A, \t}} \longmapsto Az + \tau\ \qquad \text{with}\ \det A > 0\ .$$ 
  The RP $\gamma _{A, \tau} $  of a Gauss   filter  determined by the transformation $T_{A, \t}$  is   
   $$ \g_{A, \t}(z) \=   \frac{1}{\det{A}}  \frac{1}{\sqrt{2\pi}} e^{-\frac{1}{2}\big\vert A^{-1}(z -\t)\big\vert^2}\ .
   $$
   Clearly, the     group $\operatorname{Aff}^+(\bR^2) $  acts    transitively   on    the      space $\mathcal{G}$
          of   all Gauss filters and its
      stability subgroup is $\SO(2)$.
 This   action  canonically  extends to an  action
   of the larger group $\SL(3, \bR)$,   so that $\mathcal{G} = \operatorname{Aff}^+(\mathbb R^2)/ \SO(2) = \SL(3,
   \bR)/\SO(3)$  \cite{L-M-R}.\par
   \medskip
      The  orbit  in $\cG$ of the mother Gauss filter under the  {\it similarity  subgroup}
      $$\Sim(\bR^2)= \CO(2) {\cdot} \bR^2 \simeq \bC^* {\cdot} \bC $$
  is  an important    submanifold  $\mathcal{G}_0 = \Sim(\bR^2)/\SO(2)$ of the manifold of all  Gauss filters.  If 
  we identify  $\bR^2$  with the  plane of complex numbers  $\bC = \{z = x +i y\}$, then 
  $\Sim(\bR^2)$ is identified with the {\it complex   affine   group}  $\bC^* {\cdot} \bC$ of  the complex transformations
    $$ z\overset{T_{a,b}}\longmapsto az+b\ , \qquad a \in \bC^*\ ,\ b \in \bC$$
and the  RP of the  Gauss   filter   associated with the transformation  $T_{a,b} \in \Sim(\bR^2)$  is 
  $$ \g_{a,b}(z) =  \frac{1}{\vert a\vert^2}   \frac{1}{\sqrt{2\pi}}  e^{- \frac{\vert z-b\vert^2}{2\vert a\vert^2}}.$$
 Note that   $\g_{a,b} = \g_{|a|,b}$ for any $a \in \bC^*$, meaning  that   the  subgroup $\bR^+ {\cdot}
  \bC \subset \bC^* {\cdot} \bC$ acts simply transitively on $\mathcal{G}_0$.
  The  parameter  $  \sigma = \vert a\vert $  is called the   {\it standard deviation}  and $b$ is  the {\it mean  value} of  $\g_{a,b}$.
Notice also  that, when   the    standard  deviation   $ \sigma = |a| $    tends  to $ 0$,    the Gauss  functional
 $$T_{ \g_{\s, b}}(I(z)) := \int I(z)\gamma_{ \sigma, b} \vol(z) $$
      tends   to   the Dirac    delta  function  at $b$,  i.e. to  the   functional  $\delta_{b}(I) \= I(b)$.
For this reason,  {\it the  Gauss   filters of $\mathcal{G}_0$  can be  taken  as    $\sigma$-approximations    of  the
functionals  $ I \mapsto I(b)$,  $b \in \bC$.}  \par

\medskip
 \subsection{Kuffler isotropic neurons and Marr  filters} \label{Kuffler}
   S. W. Kuffler \cite{Ku,H}   was  the  first who   detected   a  response  to   a  stimulus of  the retinal ganglion cells  in
   mammals.   He   described   the   structure  of  the isotropic (i.e. rotationally invariant) receptive field  of an
   (isotropic) neuron   as two concentric    discs  $D' \subset  D $  of the  retina and divided the isotropic neurons in two classes: the {\it ON-neurons} and the {\it OFF-neurons}.   The   receptive  profile $W(x,y)$ of an
   ON-neuron (respectively,  an OFF-neuron) is 
    \begin{itemize}
    \item positive (resp.  negative) in  the inner disc $D'$,
    \item negative  (resp.  positive)  in the  ring $D \setminus D'$
        \end{itemize}
and it is such that   $\int_{D} W(x,y) \vol = 0$. This  feature explains  why  the neurons of this kind  give  no   response   when  the input  function $I$ is constant.\par
  D. Marr    \cite{Marr}  showed  that  a linear filter  whose  RP  is the 
Laplacian  $\Delta \g_{a,b}$ of  a   Gauss  density   $\g_{a,b} \in \cG_0$
 provides  a     realistic    model  for a   Kuffler   neuron.  He also   explained  that  a  system of  such filters
 produces  a  regularisation   and    contourisation of  the input   function $I$. In  other  words, it  transforms  the
 retina  image   into  a   graphics picture.  This   is the    aim of  the data processing  in   retina.  We remark that systems of Marr filters (and of their generalisations) are used in computer vision to transform    pictures
 into graphics.\par
 Another realistic model  for a    Kuffler neuron  is  a linear filter   with RP  given by  the  difference between   two Gauss
 densities, both   with    same mean value $b$,    but  with  different standard  deviations $\s, \s'$.\par

 \par
 \medskip

\subsection{Koenderink's Multiscale Geometry of  Image  Processing}
  J. Koenderink \cite{K,F-R-K-V}    defines    the  {\it Multiscale Geometry}   as  the geometry  which  studies
   the $\s$-approximations to   Differential Geometry   for  arbitrary  resolution parameters  $\sigma$.
In   his  seminal paper \cite{K}, he  showed   that an image   can be  embedded  into  a one-parameter
  family of  derived  images, parametrised  by   the  resolutions $\s$ and   governed by  the heat   equation   or   other
  types of  diffusion equations,  in particular the anisotropic ones (the latter is particularly relevant in computer vision,
  as they do not  diffuse  edges).
\par
\medskip
\subsection{Derived   filters  and  Hansard and Horaud's simple  cells  of order  $k$} \label{derivedfilters}
 Let $X$  be   a vector  field on $\bR^2$, identified with  a derivation  of  the  algebra of  real functions.   Given a
 Gauss filter $T_{\g_{a,b}}$, the linear functional   $T_{X {\cdot} \g_{a,b}}$   with RP   $X {\cdot} \g_{a,b}$
  is  called     {\it derivative  of $T_{\g_{a,b}}$  in the  direction  $X$}.
 Integration   by  parts   shows    that the   limit  of  $T_{X {\cdot}\g_{a,b}}$  for    $\sigma = \vert a\vert  \to 0$
 is   the functional  which associate to any input function $I$ its directional derivative  $ -(X {\cdot} I)(b)$ at the
 point $b$.
In  differential geometry,    such a  functional is identified with  the tangent vector $- X_b$. For this reason, the
functional  $T_{X {\cdot}  \g_{a,b}}$  can be considered as    a    {\it  $\sigma$-approximation   of   the   tangent
vector  $-X_{b}$}. Similarly, the   functional   $T_{Y{\cdot}(X{\cdot} \g_{a,b})}$  is  a {\it $\sigma$-approximation of
the    second order differential operator at $b$ given by}
$$I \longmapsto Y{\cdot} (X{ \cdot} I)(b)\ .$$
M. Hansard and R. Heraud  \cite{H-H}  proposed  a   definition  of    {\it  a simple visual neuron  of  order  $k$}   as
the    filter with   RP  given by a linear combination of  directional derivatives of the form
    $$X_1 {\cdot} X_2 {\cdot} \ldots {\cdot}  X_k {\cdot} \g_{a,b} \  .  $$
As we remarked above, such an operator can be  considered as    a $\s$-approximation   of a linear combination of    differential
operators  of the form $(-1)^\ell X_1\circ \cdots \circ X_k\vert_{b}$.
    This means that, geometrically,    a simple  visual  neuron of  order  less than or equal to  $k$   computes a component of
    the   $k$-th order  jet   of   a contour.     We recall   that the  space   of  all $k$-jets  of  non-parametrised  curves
    of the  retina  $ R \simeq \bR^2$ are  locally  parametrised  by the  space  of $k$-jets  $J^k(\bR, \mathbb{R}^2)$  of the
    $\bR^2$ valued  functions  on  the  real line,  i.e.  the   space of  the Taylor polynomials  of functions of one real  variable and values in $\bR^2$ of  degree   less
    than or equal to $k$.\par
    Following an  idea by   Hubel and  Wiesel,   Hansard and  Horaud     also   proposed    the notion  of  a {\it complex
    visual cell}  as a composition of   simple  visual neurons of the above kind.
\par
\medskip
\subsection{Gabor   filters   and  simple  cells of   V1  cortex} \label{sect3.5}
   Roughly speaking,   a     {\it Gabor    filter}  is  a linear functional  with  RP    given by a  Gauss  function
   modulated  by $\cos y$  or  $\sin y$ (see  \cite{D,Sz} and references therein). More precisely, it is defined as follows.
   Let  us denote by  $\Gabor_{\g^+_0}$  and   $\Gabor_{\g^-_0}$   the linear filters    with    RPs
       $$  \g^+_0(z):= \g_0(z)\cos y = e^{-\frac{1}{2}\vert z\vert ^2} \cos y\ ,\qquad  \g^-_0(z):= \g_0(z)\sin y   = = e^{-\frac{1}{2}\vert z\vert ^2} \sin y\ ,   $$
  respectively (here,  as usual, $z = x +iy$). These RPs  conveniently combine into  the   {\it complex} RP
        \beq \label{mothergaborcomplex}  \g_0^{\bC}(z) := \g_0(z) e^{iy} = e^{-\frac{1}{2} \vert z\vert ^2 +iy}  =   \g^+_0(z) + i  \g^-_0(z) \ . \eeq
   The  complex  filter  $\Gabor^{\bC}_{\gamma_0}$    with   RP
$\g^{\bC}_0$     is called    {\it complex mother  Gabor  filter}, while  its   real  and imaginary parts  $\Gabor_{\gamma_0^+} $  and
  $\Gabor_{\gamma_0^\bC} $     are called   {\it  even  and  odd  mother Gabor filters}.
      In analogy with the definition of the Gauss filters, the {\it Gabor   filters}   are the linear   filters
      $\Gabor_{\g_{a,b}^\pm}$ that are  obtained   from  the mother  filters  $\Gabor_{\g^\pm_0}$   by   
      transformations $T_{a,b} \in \Sim(\bR^2) = \bC^* {\cdot} \bC$.\par
    The RP $ \g^\pm_{a,b}$  of  these linear  filters  can be explicitly given determined   from  $\g_0^\bC$ as
    follows.
   Let  $a \= \s e^{i \theta}$,  $b_0 \= \Re(b)$, $b_1 \= \Im(b)$. Then,  by  \eqref{density},  the RP $ \g^\pm_{a,b}$
       are the  real and imaginary parts of the complex RP
         \begin{multline} \label{tranG}
         \g^\bC_{a,b}(z) = \g^+_{a,b}(z) + i \g^-_{a,b}(z)
          = T_{a,b}(\g^\bC_0(z))  =  \frac{1}{\sigma^2} e^{-\frac{1}{2} \vert z\circ T_{a,b}^{-1}\vert ^2 +iy\circ T_{a,b}^{-1}}  =\\
         = \frac{1}{\sigma^2}e^{-\frac{\vert z-b\vert ^2}{2\sigma^2} + i \frac{- (x - b_0)\sin \theta  + (y - b_1) \cos
         \theta }{\s}}\ .
         \end{multline}
Since  $\Sim(\bR^2) = \bC^* {\cdot} \bC$ acts  simply transitively  on the  family  of  all Gabor  filters of
the plane,  these  filters constitute  a manifold $\mathfrak{Gab}$  which we  identify with 
$ \mathfrak{Gab} =  \bC^* {\cdot} \bC$. \par
 J. D. Daugmanm   showed  in  \cite{D} that there exists  an uncertainty relation  between    orientation and    spatial frequency of stimuli  and   that  a   Gabor   filter  optimises  such uncertainty  relation.\par
 \smallskip
        Other mathematical   models  for   simple   neurons   is given by the  derived filters  of  first  and  second orders of the Gabor filters
        in the sense of  Hansard and Heraud, that is the analogs of the derived filters considered   in \S  \ref{derivedfilters}.
            In fact,  these   types of  filters   are  very close one to the other  and, roughly speaking,   they are both appropriate  for   detecting   second order
            jets  of  contours.\par
\medskip

      \section{The functional  architecture  of   the  V1 cortex in  statics}
      \label{sect:3}

 \subsection{Structure of  the V1  cortex}
In this section we  briefly recall  some  well  known facts  on the structure of the V1 cortex.
 For  a much more  detailed   discussion and further information, the reader is referred to the   excellent    book by D. H. Hubel \cite{H}.\par
 \smallskip
The primary visual cortex  V1  is  a layer which   is  in average 2.5 mm  thick and  consists of   six  sublayers.   Hubel
and Wiesel  discovered  that  the neurons of the V1 cortex  are organised into  vertical columns.  Each  column  consists of $80-100$ visual
neurons (of which  approximately  25\%  are  simple  cells)   and  all of them have   almost the  same  receptive field (RF). There are two type of  columns.
A column   is called  {\it regular} if  its  simple neurons  have  almost  the same orientation, say $\theta_0 \in [ -\frac{\pi}{2}, \frac{\pi}{2})$.
This means that  each of them  fires only when   a   contour crosses   their  RF  with the    orientation  $\theta_0$
  (up to an error of  $15-20\% $).
  A column  is called {\it singular}  (or   {\it pinwheel}) if  it contains  simple cells that  can
    detect contours with {\it any} orientation.\par
    \smallskip
According to this,  to  each  regular column  one  may    attribute:   (a) the point $z$ of the retina $R$, which is the common RF of its simple cells and (b)
the  common fixed  orientation $\theta_0$  of its neurons.  On the other hand, to each pinwheel one may attribute  the common RF  $z \in R$ of its neurons and 
consider   the   orientations of  the simple cells of the column as the values of an  internal parameter, which  distinguish each neuron from the others. 
In other words,  the correspondence between points of the retina  and orientations $\theta$ is a well defined  function if it is restricted to  the subset of  the  points  $z \in R$ that are  the RFs of regular columns, but such a function  has 
singularities at the points that are RFs  of  pinwheels.
   \par
        \smallskip
    The  distance  between   two neighbour pinwheels in  the visual cortex  is in average equal to $1-2$ mm.
       They serve  as a sort of watch towers, which  can detect  contours  of   arbitrary orientations.
   One of  the  purposes of  the fixational eyes movements is  to produce   shifts of  the retina  images in  such a way that
   the contours  of such images can  intersect  the RFs of the  pinwheels, so that the pinwheels  can       detect  them.  \par
\par
\medskip
\subsection{The V1 cortex  as  a  fiber  bundle: The  ``engrafted  variables" of Hubel and Wiesel and  the hypercolumns}\label{hubell}
Hubel  and   Wiesel  remarked that   the  firing  of  the simple neurons  depends not only  on  the  coordinates
$z=(x,y)$   of
 their  RF in the retina $R$,  but  also  on  many other   parameters,  called {\it internal}, which   may be  considered  as fiber-wise
   coordinates of   a   fiber bundle over  the retina. In fact, they  proposed  {\it the  first fiber bundle  model  for the
   V1 cortex}, the so called  {\it ice-cube model}   with  two internal parameters:  the orientation   and   the ocular
   dominance.\par
    In  \cite{H}  Hubel wrote:  "What the cortex does is to map not just two but
many variables on its two-dimensional surface. It
does so by selecting as the basic parameters the
two variables that specify the visual field coordinates
(distance out and up or down from the fovea),
and on this map it engrafts other variables, such as
orientation and eye preference, by finer subdivisions."\par
 The complete   set of  the  internal parameters,  which  are relevant  for  the visual  system,  is  not known.
The  most  important internal   parameters  are  the {\it  orientation}   and  the
{\it    spatial  frequency}, but there are many other relevant  parameters,  as for instance the  parameters of  the {\it
color   space}, the {\it contrast},  the {\it curvature},  the {\it temporal  frequency}, the {\it ocular  dominance},
the {\it disparity} and the {\it direction  of the  motion}.\par
Another  fruitful   idea by Hubel and Wiesel   is the  notion  of {\it hypercolumn}.
They  define  a hypercolumn as a collection of  columns consisting of  simple neurons
which  can  measure  any   possible  value   for   the   internal parameters.
Unlike the columns,  the existence of hypercolumns is not  physiologically confirmed
 and it has been the subject of debate for a long time. 
 \par
\medskip

\section{Eye movements   and the  configuration space of the  eye }\label{sect 4}
\subsection{The eye as a rigid rotating ball}\label{subsect  4.1}
>From  a mechanical point of view, the  eye  is a rigid  ball  $B_{\text{eye}}$,  which can  rotate around its center  $O$.
The  retina approximately occupies      two thirds   of the  eye  sphere $ S^2_{\text{eye}}= \partial B_{\text{eye}}$,   but,  for  simplicity,  in what follows  we  identify  it  with  the
whole eye sphere. As we mentioned in \S \ref{sec:1}, in this paper we   assume
that  the  {\it eye nodal point}   (or {\it optical center}) $\cN $ is  in  the eye sphere  and   that  its opposite
point $\mathcal F \in S^2_{\text{eye}}$  is the center of  the  fovea.\par
\smallskip
  There is a  primary (standard) position  for  the  eye, which is  characterised   by    the  orthonormal frame
  $\left(O, (\underline i,\underline j,\underline  k)\right)$ at the center $O$ of  $S^2_{\text{eye}}$,
  with   $\underline i$  giving    the  frontal direction of the  gaze,  $\underline j$   orthogonal  to $\underline i$ and    giving the   horizontal   direction oriented  from left to right,
and $\underline k$ giving   the  vertical direction oriented  from bottom to top.  This  frame determines  the  primary
  ``fixed head centred  coordinates''  $(x,y,z)$ for the  Euclidean space $\bR^3$ with  origin  $O = (0, 0,0)$. \par
Any other position of the  eye sphere is  determined  by an  orthogonal transformation $R \in \SO(3)$ around  the point $O$ (which
clearly preserves $S^2_{\text{eye}}$) that  transforms  the orthonormal frame  $(\underline i,\underline j,\underline k)$ into
another     $(\underline i',\underline j',\underline k') = R(\underline i,\underline j,\underline k) $  in which   $\underline
i'$  gives  a  new gaze direction. We recall that  any  orthogonal transformation  $R \in \SO(3)$  is a   rotation    around some    fixed axis   through the origin  $O$. In the following,
given a unit vector   $\underline e$ and an angle $\a$, we denote by $R^\a_{\underline e}$ the rotation of angle $\a$ around the axis determined by $\underline e$.
\par

  \subsubsection{Helmholtz's physiological  definition of  a ``straight line''}
  H.  von Helmholtz  gave  the  following physiological definition of a straight line \cite{He} (see also \cite{R-B} and
  references therein):\\[5 pt]
 {\it A straight  line is a curve $\ell= \{\gamma(t), t \in \bR\}$ of the Euclidean space  $\bR^3 $, which is
 characterised by the  following property: when the  gaze moves  along  the curve $\gamma(t)$, the  retinal image  of
 $\ell$  does not change}  (Fig. 3 \& 4).\\[5pt]
     Indeed, let  $\ell =\{ \gamma(t), \, t \in \bR \}  \subset   \bR^3$  be  a    straight line   and  $\Pi$ the
(unique) plane  through  $\ell$ and the  nodal point $\cN$.
   When      the  eye looks at a point   $\gamma(t)$ of the  line    $\ell$, the corresponding   image  on the  retina
   is determined  by the  central projection  $\pi_{\cN}:  \ell    \to  S^2_{\operatorname{eye}}$     and  this  is just the
   intersection  between  $S^2_{\operatorname{eye}} $ and the plane $\Pi$. Hence such an image is always the same and it 
   does not depend on the  point $\g(t)$ of the line $\ell$ toward which the gaze is directed.   \par

 \includegraphics[width=5cm ]{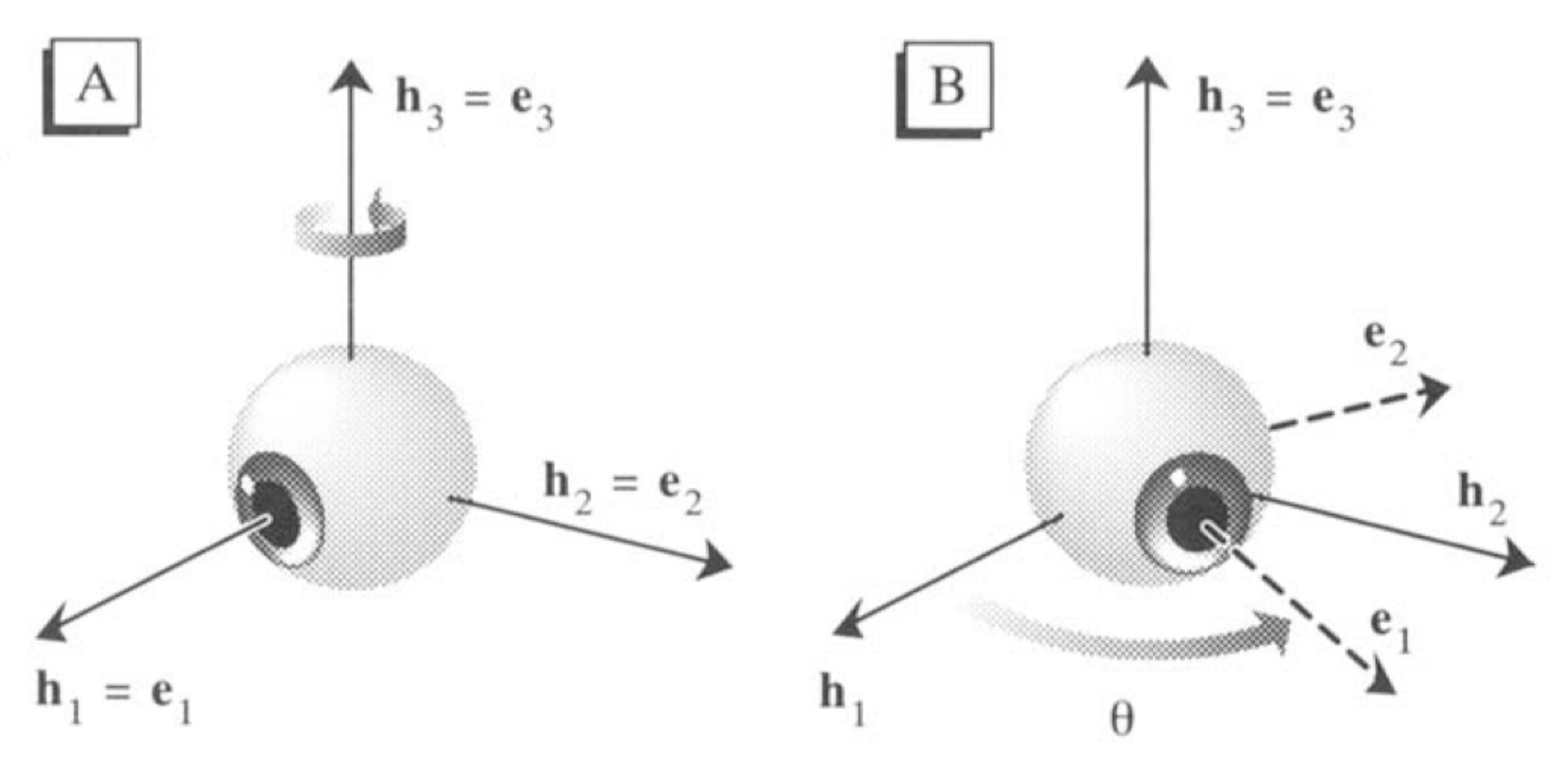} \hskip 2 cm
    \includegraphics[width=4cm, height = 2.6cm]{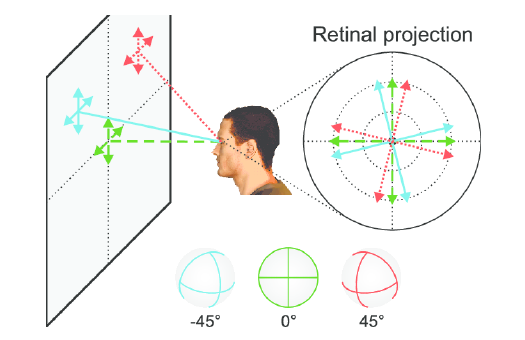}
   \ \\ [-2.5 cm]
 \begin{picture}(280,75)(0,0)
 \setlength{\unitlength}{1pt}
 \linethickness{0.3cm}
\put(47,65){\color{white}\line(1,0){20}}
\put(120,65){\color{white}\line(1,0){20}}
\put(15,10){\color{white}\line(1,0){20}}
\put(89,10){\color{white}\line(1,0){20}}
\put(110,9){\color{white}\line(1,0){20}}
 \linethickness{0.25cm}
\put(135,19){\color{white}\line(1,0){10}}
\put(135,45){\color{white}\line(1,0){7}}
\put(143,30){\color{white}\line(1,0){7}}
\put(61,31.45){\color{white}\line(1,0){20}}
\put(23,15){\tiny$\underline i$}
\put(142,15){\tiny$\underline i'$}
\put(95,15){\tiny$\underline i$}
\put(123,65){\tiny$\underline \omega = \underline i \times \underline i'$}
\put(110,10){\tiny$R^\alpha_{\underline \omega}$}
 \end{picture}
 \par
   \centerline{\bf \small Fig. 3 \& 4 (from \cite{BDL})-- Listing's Law and Helmholtz's definition of ``straight line''}
%

\subsection{The configuration  space of   the  eye -- Donders' and Listing's laws}\label{subsect  4.2.}
{\it  Donders' law}  states  that when   the head is  fixed,  the position of the eye is  completely determined  by the  unit
 vector  $\underline i'$ that gives the direction of the gaze.  This implies that the  set $\Sigma $  of the possible
 positions   of  the  eye  is a  surface  in  the  orthogonal group $\SO(3)$ (corresponding to the  set of all admissible
 gaze directions  $\underline i' \in S^2$). More precisely, we have
\par
\begin{theorem}[Donders' ``no twist''  law]  If the head is  fixed, the  direction   $\underline i'$ of  the gaze
determines  the position  of  the eye ball and does not depend on the previous eye movements.
\end{theorem}
\noindent This law implies  that there is a (local) section
$$s : S^2 \longrightarrow \mathcal{OF}(S^2)  =\SO(3)$$
 of the bundle of the oriented orthonormal
frames
$$\mathcal{OF}(S^2)  = \SO(3) \longrightarrow S^2= \SO(3)/\SO(2)$$  over  the  sphere $S^2$ of unit vectors, which
transforms any   curve of gaze directions
$$t \longmapsto \underline i(t) \in S^2$$
into  the     curve
$$t \longmapsto s(\underline i(t))\in \SO(3) = \mathcal{OF}(S^2) \ ,$$
 given by the rotations $R(t) = s(\underline i(t))$ of the eye ball associated with   the gaze curve    $\underline i(t)$.
The   {\it Listing's
law}   determines this surface as follows.
\begin{theorem}[Listing’s law -- see e.g.  \cite{A}] The movement  from the standard   eye  position,    associated  with
the  gaze    direction $\underline i \in S^2_{\text{\rm eye}}$  to  a new position,  associated with  a  gaze $\underline
i' \in S^2_{\text{\rm eye}}$,  can be always realised    by the rotation $R^{\alpha}_{\underline \omega}$   by  the angle
$ \alpha = \wh{\underline i , \underline i'}$       around    the axis  $\bR\underline \omega $ where $\underline \omega$ is the unit vector $\underline \o
\= \underline i \times \underline i'$,   belonging  to the   {\rm  Listing    plane} $\underline
i^{\perp}$.\par
 More generally,   the  movement  from a  gaze direction  $\underline a \in S^2_{\operatorname{eye}}$ to a   gaze
 direction  $\underline b \in S^2_{\operatorname{eye}}$  is realised by   the  rotation $R^{\alpha}_{\underline{\omega}}$
 by the  angle
$\alpha = 2 \wh{(\underline i+\underline a), (\underline i+\underline b)}$  around  the axis  $\bR \underline  \omega$,
$\underline \omega \=( \underline i+\underline a) \times(\underline i+\underline b) $.
The  trajectory of the gaze is  the   arc  $\overset{\frown}{a  b}$  in the (not necessarily great)   circle
 $$S^1_{\underline a,\underline b}  = S^2_{\operatorname{eye}} \cap \Pi ( \underline a ,\underline b, - \underline i)$$
given by the intersection  between the  eye sphere and  the  plane
 $\Pi ( \underline a , \underline b, - \underline i)$  through the points  $\underline a,\underline b$ and  $-\underline i$.   \end{theorem}
  Donders' and Listing’s  laws imply      that the configuration space of  the  eye positions   is  the  surface $\Sigma
  \subset \SO(3)$  made of the rotations  $R^{\alpha}_{\underline{\omega}}$
 around the axes  $\bR \underline \omega$ lying in    the Listing plane
   $ \underline i^{\perp}= \text{span} ( \underline j,\underline k)$.  A   description   of  the configuration   space of
   the  eye in terms of the Hopf bundle is  given in \cite{A1}.
  \par

\subsection{Eye movements: the   saccades   and the   fixational  eye movements  }\label{4.3}
The   eye movements    are very important  for    vision:
Classical    experiments  by  A. Yarbus  show that a  compensation of all   eye movements  leads to a loss  of  vision of stationary objects  in   $2-3$  sec.
\par

One of the possible reasons for this   is  that the  firing of
many  visual neurons    depends on  the   rate of  change of  the  internal parameters of the  stimuli.
In  this    subsection  we shortly   discuss some special eye movements that  occur when the gaze is  ``fixed", i.e.
 concentrated on  a  stable  external point:  the  {\it saccades}     and the 
 {\it fixational  eye movements}  (see e.g. \cite{R-P}).
\par
\smallskip
 The {\it saccades}  are very rapid rotations of the eyes  (up to
$700^\circ$/sec  in humans) of large
amplitudes.  It is  a ballistic  movement:  After  the start, the oculomotor  system   cannot   change  its  trajectory.
  For  any  two  admissible (i.e. belonging to the same field of   vision) gaze  direction  $\underline a,\underline b    \in S^2_{\text{eye}}$, there  exists  a unique  saccade     which    change  the  gaze  direction  $\underline a$  to $\underline b$ and   it is  realised  by   the  rotation  around  the   axis, which is   orthogonal  to the  plane $\Pi ( \underline a , \underline b, - \underline i)$,   as it  describes  above in the Listing's law. The     corresponding   evolution of the gaze  describes  a geodesic  of $ S^2_{\text{eye}}$   if and only if    both $\underline a,\underline b$  lie in a  meridian  of $S^2_{\text{eye}}$  through  the poles  $\underline i$ and $-\underline{i}$.\par
\smallskip
The {\it fixational
eye movements}  include the {\it tremor}, the {\it drifts} and the {\it microsaccades}.
The
{\it tremor}   is an aperiodic  wave-like motion of the eye of high
frequency  ($40-100$ Hz) but  of  very small amplitude  (in between a few arcsecond and  a few arcminutes) \cite{A-A}.
Under the  tremor  the  gaze fills  a cone  in one tenth of a second.
\smallskip
The {\it  drifts} occur simultaneously
with the tremor and are slow motions of the eye,  with frequency $1-20$ Hz  and  amplitudes  up to  $10'$. During   a drift,
the image of a fixation point   stays within the fovea.
 The  {\it   microsaccades}  are  fast  short   jerk-like  movements  with amplitude $2' - 120'$  (see e.g. \cite{P-R}).
\par\smallskip
 In $1$ second the tremor moves the retinal image of  $1 - 1.5$
diameters of the cones of the central part of the fovea, a  drift moves along $10 - 15$ diameters of the  cones and the
microsaccades carry an image
across dozen to several hundred diameters of the cones.\par
\smallskip
\subsubsection{The  purposes of the saccades and   the fixational  eye movements}\hfill \par
(i) When   the  eye   is    fixed, the
 visual information comes from the  light rays   passing through  the lens of  the eye, which    can be identified
 with the  rays   of the central projection  by the nodal point $\cN$  -- see  \S \ref{sec:1}, n.3.
  These rays determine
 a  $2$-dimensional   submanifold
$\text{L}(\cN) \subset \text{L}(\bR^3) $ of the space $ \text{L}(\bR^3)$ of all lines of $\bR^3$.
Thanks to the fixational eye movements, the brain receives information from a neighbourhood of $\text{L}(\cN)$  in   the
four dimensional manifold $ L(\bR^3) $.\par
(ii) The fixational eyes  movements  produce  also  shifts of  the  retinal image, which allow    the contours to
intersect  the  receptive  fields of the pinwheels, so that they can be    detected  by the simple cells  of those
pinwheels.
These  movements  are essential for perceiving   immobile objects. \par
(iii) We conjecture that the  main aim of   the tremor is to increase the width of the retinal  contours,  so that they   can be
detected by several rows of photoreceptors. This allow  to estimate the gradient of the input function at   the points
of the contours.
\par

\medskip

\section{A short introduction to the conformal geometry of the sphere}\label{section  5}
\subsection{The conformal sphere and its group of conformal transformations}\label{subsection  5/1/}
In this subsection, we recall  some basic  notions  and  results  of the  conformal   geometry    of the 2-sphere. For
more extended  discussions and many additional  references for this topic, see e.g. \cite{Ko, St, A-G,P-S,Su}.\par
\smallskip
\subsubsection{The conformal sphere}  \label{confsphere} The {\it conformal  structure}  of a Riemannian  manifold $(M, g)$ is    the   class   $[g] = \{ \lambda g\
: \   0< \lambda  \in C^{\infty}(M)\}$ of  all metrics that are  conformally  equivalent to $g$. The  {\it conformal
group   of $(M, [g])$} is  the   group     $ \operatorname{Conf}(M)$  of the   transformations,   which   preserve the
conformal structure $[g]$. 
\par
The {\it conformal sphere} is the sphere $S^2 = \{\ x^2 + y^2 + z^2 =  1\}\subset \bR^3$  equipped with the conformal
structure $[g_o]$,  determined by  the standard round metric  $g_o$.  The {\it conformal plane} $(\bR^2, [g^E])$ is the plane $\bR^2$, 
  equipped with the conformal structure $ [g^E]$ determined by the standard Euclidean metric $g^E$ of $\bR^2$. \par
\smallskip
Let  $\Np$ and $\Sp$ be  the
north pole  $\Np = (0, 0, 1)$   and  the
south pole  $\Sp = (0, 0, 1)$ of   $S^2$, respectively,  and let
$  \ster_{\Np} : S^2 \setminus \{\Np\} \longrightarrow  \bR^2$ be the  stereographic projection with respect to $\Np$ from  $S^2 \setminus \{\Np\}$  onto the   tangent plane $T_{\Sp} S^2 =\bR^2 $  at the south pole   $\Sp$.    It is a conformal  diffeomorphism between the conformal sphere minus the north pole  $(S^2 \setminus \{\Np\}, [g_o])$ and the conformal plane $(\bR^2, [g^E])$.  Such 
a diffeomorphism has a  unique  extension to  a diffeomorphism $\f: S^2 \to \wh \bC = \bC \cup \{\infty\}$ between   $S^2$ and  the Riemann   sphere $\widehat \bC =  \bC \cup  \{\infty\}$,   mapping  $\Np$
  to  $\infty$.  This diffeomorphism  can be used to induce a  Riemann metric  $g_o' \= \f_* g_o$  and a corresponding  conformal structure $[g'_o]$ on $\wh \bC$. In this way,   $\f$ is  a conformal equivalence between $(S^2, [g_o])$ and $(\wh \bC, [g_o'])$. \par
 \smallskip 
Identifying $S^2 \setminus \{\Np\}$ with $\bR^2$ by means the stereographic projection $\ster_{\Np}$,  the  complex   coordinate   $z =x +iy $ of $ \bC$  can be considered as   a complex 
  coordinate  on       $S^2\setminus \{\Np\}$.    It  is called     {\it standard (local) holomorphic coordinate}  for the $2$-sphere.
  \par
 \subsubsection{The conformal group of  $S^2$} \label{Gaussdecomp}The group  
$\operatorname{Conf}(S^2)$  of  the conformal transformations of $(S^2, [g_o])$ has two connected components.
The    universal covering of   the connected component of the identity $ \operatorname{Conf}^o(S^2)$ of  $\operatorname{Conf}(S^2)$  is    $ \SL(2,\bC)$  and     
 \beq \label{8}    \operatorname{Conf}^o(S^2) = \SL(2,\bC)/\{\pm I_2\}    \simeq \operatorname{SO}^o(1,3)\eeq
(see   also  \S \ref{subsect:5.2}). 
  By means of the conformal equivalence  $\f: S^2 \to \wh \bC$,   each    $f \in  \operatorname{Conf}^o(S^2)  $ corresponds to   the conformal transformation $\wh f \=  \varphi \circ f \circ \varphi^{-1}$ of the Riemann sphere $\wh \bC$. 
The   transformations $\wh f$  are  the   {\it linear fractional transformations}, i.e.  the transformations of the Riemann sphere  of the form
\beq \wh f(z) =    \frac{az + b}{cz +d}\ ,\ \qquad    A  = \left( \begin{array}{cc} a & b\\ c & d\end{array}\right)\in \SL(2,\bC) \ .
\eeq
The  map    
$$A = \left( \begin{array}{cc} a & b\\ c & d\end{array}\right)  \in \SL(2,\bC) \longmapsto \wh f_A(z) \=  \frac{az + b}{cz +d}  \in \operatorname{Conf}(\wh \bC)$$
 is a group homomorphism and its kernel contains  exactly  two elements,      $\pm I_2$.    it follows that  
 is
$
\operatorname{Conf}^o
(\wh \bC) \simeq \SL(2, \bC)/\{\pm I_2\}
$,   as  we claimed above. \par
\smallskip
Consider now the subgroups of $\SL(2, \bC)$  defined by    
\beq\label{gd}  \begin{split} & \hskip4  cm  \bC^* = \left\{ \begin{pmatrix}
a&  0\\
0& a^{-1}  \end{pmatrix}, \, a\in \bC\setminus\{0\} \right\} \\
& N^+  =  \left\{  A   =
\begin{pmatrix}
1&  b\\
0& 1 \end{pmatrix}\ ,\ b \in \bC\ \right\}\qquad \text{and}\qquad N^-  =  \left\{  A   =
\begin{pmatrix}
1&  0\\
c& 1 \end{pmatrix}\ ,\ c \in \bC\ \right\}\ , 
\end{split}
\eeq
i.e.  the  diagonal subgroup and the   unipotent  subgroups of  upper  and  low  triangular  matrices of $\SL(2, \bC)$, respectively.
Let  also denote by  
$$B^- \=  N^- {\cdot} \bC^*\ ,\qquad B^+  \= \bC^*  N^+$$
 the   Borel     subgroups of   $\SL(2, \bC)$ given by 
    the  lower  triangular and the upper  triangular    matrices, respectively.
Notice that each matrix $A$ of   the  open  and dense   subset   
$$\cU = \{ A =\left( \smallmatrix a & b\\ c & d \endsmallmatrix\right) \ :\ a \neq 0\}  \subset \SL(2, \bC)\ , $$
can be decomposed into  the product 
\begin{multline*} A = A^-  A^0 A^+\ ,\\
 \text{with}\   A^- = \left(\begin{array}{cc} 1 & 0\\ \frac{c}{a} & 1 \end{array} \right)\in N^{-} \ ,\     A^0  = \left(\begin{array}{cc} a &0 \\ 0 & \frac{1}{a} \end{array} \right)\in \bC^* \ ,  \ A^+ = \left(\begin{array}{cc} 1 & \frac{b}{a}\\ 0 & 1 \end{array} \right)\in N^{+}\ ,\
\end{multline*}
meaning that  $\cU \subset \SL(2, \bC)$  admits the  {\it Gauss decomposition}  
$ \cU =  N^- {\cdot} \bC^* {\cdot} N^+$.\par
  \smallskip
 We conclude this short subsection with a lemma in which we  summurize  some known  facts  on the  transformations in  $\operatorname{Conf}(S^2) = \operatorname{Conf}(\wh \bC)$, which  are direct consequences of the above descriptions of the  conformal transformations in terms of the linear fractional transformations of the Riemann sphere.
  In the statement,  for any given point $z \in S^2 = \widehat \bC$,  we     denote    by  $(\SL_2(\bC))_z$ the stability
  subgroup   $\SL_2(\bC)$ at  
    $z$ and by $j_z : (\SL_2(\bC))_z \to \operatorname{GL}(T_zS^2)$ the corresponding  isotropy representation.
 \par
  \smallskip
\begin{lemma}  \label{Lemma}\hfill\par
    \begin{itemize}[leftmargin = 20pt]
\item[i)] The  diagonal subgroup   $\bC^*  \subset \SL(2, \bC)$ lies   in the intersection $(\SL_2(\bC))_0 \cap (\SL_2(\bC))_\infty$ of the stability subgroups at the points  $0, \infty \in \wh \bC$  ( =   points corresponding to  $\Sp, \Np \in S^2$) and the  isotropy representations    at such points are 
$$  j_z(\bC^*)=  \operatorname{CO}(2) =\bR^+ \cdot  \SO(2)\ ,\qquad z = 0, \infty\ . $$
Further, $\bC^*$ acts on 
$ \bC = \widehat \bC \setminus \{\infty\}$ as   the   linear  similarity group    $\CO(2) = \bR^+  {\cdot} \SO(2)$.  
\item[ii)]  The stability   groups $(\SL(2, \bC))_0 $ $(\SL(2, \bC))_\infty$  at the points $0$  and $\infty$ are    the   Borel     subgroups
$$(\SL(2, \bC))_0 = B^- =  N^- {\cdot} \bC^*\ ,\qquad 
 (\SL(2, \bC))_\infty = B^+  = \bC^*{\cdot}  N^+$$
The subgroups $N^-$ and $N^+ $ are in the kernels of the two isotropic representations, while  the  isotropy   groups  $j_0(B^-)$,  $j_\infty(B^+)$ are  both isomorphic to  $\operatorname{CO}(2)  =\bR^+ \cdot  \SO(2)$.
\item[iii)]  The  group  $B^+  =  \bC^* N^+ =  (\SL(2, \bC))_\infty $   (resp. the subgroup  $N^+ \subset  (\SL(2, \bC))_\infty$) acts   on  $\bC  = \wh \bC \setminus \{\infty\} \simeq S^2 \setminus \{N\}$
as  the     group $\operatorname{Sim}(\bR^2)$  of the similarities (resp. the group $\bR^2$ of translations)  of the plane: 
$$ z \longmapsto   az +b \qquad (\ \text{resp.}\ z \longmapsto z + b\ ) , $$
\end{itemize}
\end{lemma}
We finally remark  that,    since  $j_0(N^-)  = \{\Id\}$,    {\it on a sufficiently small neighbourhood
$\cU_S \subset \wh \bC \simeq S^2$ of   $0 $,  each transformation  of  $N^-$ acts    almost as  the identity map.}\par
\medskip
 \subsection{The canonical Cartan connection of the conformal sphere} \label{Cartanapproach} Let  $\cF^o$   be  the  standard  frame   of the tangent space  $T_0 \bR^2 \simeq \bC$  of the plane at   $0$, i.e. 
 $$\cF^o =  (0,  (f_1^o = (1,0), f_2^o = (0,1)))$$
  We call {\it oriented
 conformal    frame   at a   point  $z \in \bC $} an  orthogonal  frame  $\cF =(z, (f_1,f_2))$  for   $T_z \bC$, which has   the
 same   orientation  of   $\cF^o$   and it is given by  two vectors $f_1,   f_2$   of    equal  length.   A  conformal  frame
 $\cF$ at    $z \in \bR^2=\bC  $   can  be also identified with  the  1-jet  $j^1_0(h)$  of the unique  conformal
 transformation $h  \in  \operatorname{Conf}(\bR^2) $  that  maps   $0$  into  $z =
 h(0)$ and   $\cF^o$ into  $\cF$.\par
 \smallskip
 As  the isometry  group $\operatorname{Iso}(\bR^2)$  of $(\bR^2, g^E)$  acts simply transitively on total space of  the bundle $\mathcal{OF}r(\bR^2) \to \bR^2$ of the  orthonormal frames of $(\bR^2, g^E)$, 
 the   similarity  group
$\operatorname{Sim}(\bR^2)$ acts   simply transitively on the total space of the   bundle    $\mathcal{CF}r(\bR^2) \to
\bR^2 $ of all   oriented    conformal  frames  of $\bR^2$ and   there exists a
$\operatorname{Sim}(\bR^2)_0$-equivariant  diffeomorphism
$$     \operatorname{Sim} (\bR^2) \ni  h \longmapsto      \cF_h  \=   \bigg(h(0),  (h_* f_1^o, h_*f_2^o)\bigg) \in \mathcal{CF}r(\bR^2)\ . $$

 The above  definition  of conformal frames and of 
 the bundle $ \mathcal{CF}r(\bR^2)\to \bR^2 $  can be   generalised  to the     case of  the conformal sphere as follows.   
A  {\it second order conformal frame at a point $ p \in S^2 $} is  the  $2$-jet
$ j ^2 _0 \left(h \circ \f^{-1}\right)$   at $0$ of  the composition between the conformal equivalence $\f^{-1}: \wh \bC \to S^2$  and a conformal transformation  $h \in
\operatorname{Conf}^o(S^2)$  mapping  $\Sp$ to     $p = h(\Sp)$.
 The connected  conformal group    $G= \operatorname{Conf}^o(S^2)$ acts simply transitively   on the   manifold  $
 \mathcal{CF}^{(2)}(S^2)$ of all second order conformal frames at the points of $S^2$ and the   principal bundle  
$$\pi:  \mathcal{CF}^{(2)}(S^2) \longrightarrow S^2$$ is identifiable with the   homogeneous bundle
$$ \pi:  G = \operatorname{Conf}^o(S^2) \longrightarrow      G/ G_{\Sp}  =  S^2$$
where  we denote by  $G_{\Sp}    \simeq \Sim(\bR^2)$  the  stabiliser  of the   south pole  $\Sp$.
The   left  action    of     $G  = \operatorname{Conf}^o(S^2) $ on itself   defines  a   left
invariant absolute parallelism   (i.e.  a flat  linear  connection)    on $G$, which   identifies    each  tangent   space
$T_g G$  of $G$  with  the Lie  algebra     $\gg:= T_eG  =  \mathfrak n^- + \mathfrak{co}(2)   + \mathfrak n^+$.     It is
called  the {\it canonical Cartan connection   of    the  conformal structure} of $S^2$.
 \par
\medskip
\noindent
{\it Remark.}   For  any   $n$-dimensional manifold $M$ of dimension $n \geq 3$,   with a  conformal structure $[g]$,   E. Cartan
   introduced     a canonical  principal     bundle of  second order  conformal  frames
$$\pi: P \longrightarrow   M  = P/ \operatorname{Sim}(\bR^n) $$
equipped with  a {\it Cartan connection}   (i.e. a   flat    linear connection on the total space   $P$ 
with  certain special  invariance properties with respect to the action of the  structure group $\operatorname{Sim}(\bR^n) $). 
If the dimension is $n = 2$, for a generic $2$-dimensional conformal manifold there is no such canonical principal bundle with a  Cartan connection 
 (it is a consequence of the fact that
  the   Lie  algebra of   the conformal   vector  fields   in the generic case is  
  infinite-dimensional).  However the conformal sphere  $(S^2, [g_o])$ is a  special case:
the above discussion show that it is possible to construct
   a  natural  bundle of second order frames  $\pi:  \mathcal{CF}^{(2)}(S^2) \longrightarrow S^2$
with a  Cartan connection, exactly   as it occurs for the conformal manifolds of higher dimension. 
\par
\medskip

\subsection{The M\"obius projective  model  of  the conformal  sphere} \label{subsect:5.2}
  Consider  an orthonormal frame $(e_0, e_1,e_2,e_3)$  for  the {\it Minkowski space-time} $\bR^{1,3}$, that is the vector
  space
   $$\bR^{1,3}= \mathbb{R}e_0 + \bR^3= \{ X= x^0e_0 + x^1e_1 + x^2e_2 + x^3 e_3  = (x^0,\vec x) \}\ ,$$
  equipped with   the  Lorentzian  scalar product
  \begin{equation} \label{lorentz} g(X,Y) = -x^0 y^0 + \vec x {\cdot} \vec y\ .\end{equation}
   The {\it light   cone}   at the origin  $0 = (0, \vec 0)$  is the subset of $\bR^{1,3}$ defined by
    $$V_0 = \{X \in V, \ g(X,X)=0 \}\ .$$
Up to  a scaling, the (connected component of the identity of the)  Lorentz  group $G =\SO^o(1,3)$  has    three  orbits
in  $\bR^{1,3}$, namely:
\begin{itemize}
\item $V_T = G{\cdot} e_0 =   G/\SO(3) $  (the {\it Lobachevsky space}),
\item $V_S = G{\cdot}  e_1  =  G/ \SO(1,2) $   (the {\it De Sitter  space}),
\item $ V_0 = G{\cdot}  p = G/\Sim(\bR^2) $ with  $p = \frac12 (e_0 + e_1)  $ (the {\it light cone}).
\end{itemize}
 \smallskip
 Consider the projectivisation $\bR P^3$ of the 4-dimensional vector space $\bR^{1,3}$ and the natural projection
 $$\pi: \bR^{1,3} \setminus \{0\} \longrightarrow \bR P^3 = \bP(\bR^{1,3}) \ ,\qquad \pi(v) = [v] \= \bR v\ .$$
 Under   $\pi$, the   three $G$-orbits  on $\bR^{1,3}$ are projected onto  the following three     $G$-orbits   in  $ \bR P^3 $:
\begin{itemize}
\item $V_T$ is  mapped  onto the  unit ball   $B^3 = \pi(V_T) \simeq V_T$;
\item $V_S$ is mapped onto the exterior of the unit  ball  $\pi(V_S) \simeq V_S/\bZ_2$;
\item $V_0$ is mapped onto the  {\it projective  quadric} 
$$  \mathcal Q  = \pi(V_0) \simeq  G/(\bR^+ {\cdot} \Sim(\bR^2) )  \ , $$
which, as we will shortly see, is identifiable with the conformal   sphere. 
\end{itemize}
Any  unit   time-like vector  $e_0$ determines the orthogonal  Euclidean   hyperplane   $e_0^{\perp}$ and the corresponding translated Euclidean hyperplane $E_{(e_0)} \= e_0 + e_0^\perp$.  The  intersection
 $$  S^2_{(e_0)} =  E_{(e_0)} \cap  V _0$$
  is  the   unit  sphere  of the Euclidean space  $E_{(e_0)}$ and it  is diffeomorphic to 
   the   quadric  $Q=\pi(V_0)$.  Indeed, the      map
\beq  \label{2}   \psi^{(e_0)} :  \bR P^3 \setminus \bP (e_0^\perp)   \longrightarrow     E_{(e_0)}      ,\qquad    [p] = \bR p  \longmapsto   \psi^{(e_0)}(p)
\=  \bR p \cap E_{(e_0)} \eeq
is  a diffeomorphism,  which induces   a  diffeomorphism  between $\cQ = \pi(V) \subset \bR P^3$ and  $S^2_{(e_0)} \subset E_{(e_0)}$.  The pull-back of the round metric $g_o$ of $S^2 = S^2_{(e_0)}$  by such diffeomorphism is  a    Riemannian metric  $g^{(e_0)}$   of constant  curvature  $1$  on  $\cQ =
\pi(V_0) $.    Since   replacing  the unit   time-like vector $e_0$   by  a different  unit time-like vector  $e'_0$ determines a new diffeomorphism $\psi^{(e'_0)}$ and a  metric $g^{(e_0')}$ on $\cQ$, which is  conformally  equivalent to   $g^{(e_0)}$, 
the  conformal structure  $[g^{(e_0)}]$ on $\cQ = \pi(V_0)$, determined  by  the metric  $g^{(e_0)}$ is independent of the choice of $e_0$.  The group $G =\SO^o(1,3)$ acts on $\cQ  \simeq S^2$ as the
connected M\"obius   group  $ \operatorname{Conf}^o(S^2)$ of the conformal transformations of $S^2$.
\par

\subsection{The $\cQ$-correlation   between projective  points  and projective    planes}\label{subsection  5.3.}
   The  Lorentz metric \eqref{lorentz} determines   a one-to-one correspondence
    between the projective  points and the projective planes of  $\bR P^3$ defined by 
    \beq \label{2*}  \bR P^3 \ni [v] = \bR v \qquad\overset{1:1} \longleftrightarrow\qquad  \Pi_v:=  \bP(v^{\perp}). \eeq
  It  is called    {\it correlation  with respect
   to   $\cQ$}. 
  Under this bijection, the three  types  of  projective points $[v]$,   determined by  the  {\it timelike}, {\it
  spacelike} and {\it lightlike} vectors  $v \in \bR^{1,3}$, respectively, correspond  to three kinds of  projective planes,
  namely those not intersecting  $\cQ$,    those   that are secant to  $\cQ$ and those that are  tangent  to  $\cQ$:
 \begin{align}
\label{111}  &[n] = \bR n \in   \pi( V_T) &  \longleftrightarrow\qquad&\Pi_n   \ \text{such that}  \ \  \Pi_n \cap  \cQ  = \emptyset	
 \ ;    \\
 &[m] = \bR m \in \pi( V_S)& \longleftrightarrow \qquad & \Pi_m \  \text{such that} \ \  \Pi_m \cap  \cQ  =   S^1\ ;\\
& [p] = \bR p \in  \pi(V_0) & \longleftrightarrow \qquad & \Pi_p \ \text{such that}\ \  \Pi_p \cap \cQ = [p] \ .
 \end{align}
Note that   \eqref{2} establishes also  a bijection between  the affine planes in $E^3_{(e_0)} \simeq \bR^3$  and  the  projective planes  that are different from
$\pi(e^\perp_0)$. Under this correspondence, the affine planes of $\bR^3$ which do
not intersect $S^2$ are in bijection with  the
projective planes  $\Pi_n \neq \pi(e_0^\perp)$ with  $\Pi_n \cap  \cQ  = \emptyset$, i.e. those corresponding to the lines   $[n] \in \pi(V_T)$, $[n] \neq [e_0]$.
\par
\smallskip
 The  following lemma is well known. \par
\smallskip
\begin{lemma} \label{lemma1} The stability  subgroup $G_{[p]} = \Sim(\bR^2)$ of  a point $[p] \in \cQ = S^2$ acts
transitivity on the orbit $V_T$  and  hence   on the set
 $\big\{\Pi_n,\, [n] \in \pi(V_T) \big\}$  of  the   planes of $\bR P^3$  not intersecting   $\cQ$ (corresponding to   the affine planes of $\bR^3$ not intersecting $S^2$).
 \end{lemma}
\par
\smallskip

\section{Information  processing  in  dynamics} \label{sect:6}

\subsection{Shifts of   receptive fields   and pre-saccadic  remapping   }\label{section  6.1.}
    In a seminal paper, J. Duhamel et al. \cite{D-C-G} showed    {\it  the existence of  anticipated  shifts of    receptive
    fields} for    visual neurons   in   the lateral intraparietal area (LIP)  of   a macaque.  The same   phenomenon
    was later detected    for   neurons in  many other  visual systems, including V1, V2, V3, etc. (see
    \cite{M-C,W-W,Co}).  \par
   The  shifts of receptive fields (RF)  found  by Duhamel and collaborators can be shortly described as follows. Assume that,  before a saccade,   a neuron $n$  with   RF  $z_n$   detects  the   retinal  image
    $\overline A$ of  an  external  point  $A$  (that is, the image   $\overline A$ is  in 
$z_n$)  and  that,  after   the saccade, the same neuron detects   the   retinal image   $\overline B$
of  another point   $B$ (i.e.  in  $z_n$ there is now $\overline B$).   Experiments give evidences that   the  neuron $n$
gets   information    about   $B$ {\it approximately 80  ms  before  the  beginning of the saccade},  that is when the retinal  image
$\overline B$  is not  yet  at   $z_n$, but  at some different point $\wt z$, which is  the  RF   of   some different  neuron $\wt n$.
In other words, the neuron $n$ reacts {\it as if \underline{before} the beginning of the  saccade} its  RF field $z_n$  is changed  into the RF  $z_{\wt
n}$ of  some other  neuron. 
 Due to this,   the phenomenon is  named  {\it  pre-saccadic shift  of a  RF} (see e.g. \cite{M-C,W-W}). 
 \par
   This name  is   somehow
   misleading. In fact, each  visual   neuron has a physiologically   {\it fixed} RF  -  it is a  domain  in  the retina  which is
   physiologically connected  with the neuron. 
 On the other hand,   the eye movement which occurs during a  saccade produces  a shift of all   {\it retinal images} of the external points.  Due to this, if 
    a neuron $n$  detects  a  retinal image $\overline{A}$ of an external point $A$ at its RF  $z_{n} $ before the beginning of a saccade, 
    after  the saccade  the same neuron should  detects  a new  retinal image in the RF  $z_n$ ,  namely  the   image $\overline B_n$ of a 
point $B \neq A$.  In particular, the  detection of $\overline B$ by the neuron $n$ should occur {\it after} (and {\it not before}) the saccade. Now the  ``pre-saccadic  shift  of  the RF of $n$''   
is the name for the unexpected phenomenon that   the neuron
$n$ reacts to the stimulus  $\overline{B}$ earlier than when it should be,  that is when $\overline B$ is not yet  in  $z_n$. \par
 For a long  time, it  was suggested (see  for example \cite{M-B}) that  one of  the  neurons $\tilde{n}$, having $\overline B$ in its RF,    sends  
 information  about  $\overline{B}$   to the neuron $n$ before $\overline B$ reaches the RF of $n$. However, as it was remarked  by
 M.\ Zirnsac and T.\ Moore  \cite{Z-M}, this  explanation  looks not very  realistic, since it assumes  the existence of  a huge
 number of horizontal connections between neurons.  Zirnsac and Moore offered   a different and more plausible explanation  which  may be briefly
 summarised  as follows.
 The  global information   about the stimulus   before   the saccade     is integrated in  some  higher   order
 visual subsystem.  On the basis of  this  global  information,   the  visual system chooses  a   gaze  target  for
 a    saccade  and sends  the   information  about   the retinal image  $\overline{B}$  (detected by a neuron $\wt n$) 
  to  the neuron $n$ before the saccade starts.  After  the saccade,   the  retinal  image $\overline{B}$ does come  to  the RF $z_n$ of
   $n$ and  the preliminary information that was sent to  $n$  is  corrected.\par
\smallskip
Following  Zirnsac and Moore's idea,  it is sensible  to  believe that  also  the  microsaccades are not  random phenomena, but  controlled by higher sections of the visual system, while the drift 
is just a random process during which
 the visual cortex and the oculomotor system get  stochastic  information on the stimulus.  We think that a microsaccade
 $\text{Sac}(A,B)$  from a  gaze direction  $A $ to a  gaze direction $B$    serves to   change     the  retinotopic
 coordinates   associated  with the  initial direction $A$ to   new  retinotopic   coordinates  associated    with the
final  direction $B$. This process is called {\it remapping}. The stochastic  information about  the   stimulus   during the
 drift  after the saccade $\text{Sac}(A,B)$  is  encoded  into  the  coordinates associated with  the final gaze $B$, see
 \cite{A1}.
\par
\smallskip

\subsection{Gombrich's  Etcetera Principle  and a conjecture on the remapping}
The art historian E. Gombrich formulated  the   following  important
idea, called {\it Etcetera Principle}, about the mechanism of   the saccadic  remapping  \cite{G,D-C-G}:\\[5pt]
``The   global pattern in environments such as a forest, beach or
street scene enables us to predict more-or-less what we will
see, based on the order and redundancy in the scene and on
previous experience with that type of environment. Only  a  few  $3-4$ salient  stimuli  are contained in  the
trans-saccadic  visual memory  and  update".\\[5pt]
\indent
We now   want   to    state a conjecture   about the conformal  character  of  the remapping, which  supports  the Etcetera
Principle. But for  this,  we first need to  recall  a few fundamental facts on conformal geometry of the plane and of the
sphere and on
 the  M\"obius  projective  model of  the conformal sphere.
\par
\smallskip
As in \S \ref{subsect:5.2}, let us identify  the hyperplane $E^3_{(e_0)} = e_0 + e_0^\perp  \subset \bR^{1,3}$ with  the physical $3$-dimensional
Euclidean space $E^3 = \bR^3$. Let also identify   $S^2  = S^2_{(e_0)} \subset E^3_{(e_0)}  = \bR^3$  with the eye   sphere
$S^2_{\text{eye}}$.  As we pointed out  in \S \ref{subsection  5.3.},  the diffeomorphism \eqref{2}
  determines   identifications  between $\bR P^3 \setminus \pi(e_0^\perp)$ and 
$E^3_{(e_0)} = E_3$ and  between the   projective planes   $ \Pi_n$,   $[n] \neq [e_0]$,   defined in  \eqref{111},  and  the  affine planes  $\Pi \subset E^3_{(e_0)} = E^3$ that  are    external to $S^2 = S^2_{\text{eye}} $.
  \par
  \smallskip
 Consider now an affine plane $\Pi \simeq \Pi_n$ external to $S^2_{\text{eye}}$  and the corresponding   central projection   $\varphi_\cN
 : \Pi  \to S^2_{\text{eye}} $  defined in \eqref{central}. By the map \eqref{2*}, such central projection  corresponds to  a map
 $\varphi_\cN : \Pi_n  \to \cQ$.
A  physical rotation  $R^\a_{\underline v}$ of the eye,  through the angle $\a$  around the   axis    $\bR\underline v$, 
determines   changes of  the images  of the external surfaces. In fact,  if they are expressed  in terms of the    {\it retinotopic   coordinates} (=  coordinates that are   fixed to  the eye sphere  $ S^2_{\text{eye}}$), the  images of   the external surfaces are    transformed under the  inverse rotation   $(R^{\a}_{\underline v})^{-1}  = R^{-\a}_{\underline v}$. In particular,  from the point
of view of retinotopic coordinates,  if the eye rotates by $R^\a_{\underline v}$, the image of an  external plane  $ \Pi = \Pi_{n}$  changes  into   the image of  the  rotated plane  $\Pi' =
R^{-\a}_{\underline v}(\Pi_{n})$.
The  {\it  problem  of remapping}  is the problem of  determining a bijection between the points of  $\Pi$ and  $\Pi'$  which induces  a
simple relation between  the points of the    retina  images  $\varphi_\cN(\Pi)$ and $\varphi_\cN(\Pi')$.
\par
\smallskip
According to  Lemma \ref{lemma1},  there  is  a Lorentz  transformation  $L_{\Pi, \Pi'}$  of  the stability subgroup
$G_{\cF} \subset G = \SO^o(1,3) $ of    the  fovea $\cF \in S^2_{\text{eye}} = \cQ $, which sends  $\Pi$ to   $\Pi'$.
This leads to the following:\\[3pt]
{\bf Conjecture.} {\it   The  brain   identifies    the points of $\Pi$ with the points of  $ \Pi'$   using  the
     Lorentz transformation $L = L_{\Pi, \Pi'}$, so that   the points of the  corresponding   retina images $
     \varphi_\cN (\Pi)$, $\varphi_\cN (\Pi')$ on   $S^2_{\text{eye}} = \cQ$   are  related   by    the  associated
     conformal transformation  $L\vert_{\cQ}$ of $S^2_{\text{eye}} = \cQ$ (see \S \ref{subsect:5.2}).}\\[3pt]
 \indent 
 We remark that any  conformal transformation of $\cQ = S^2$   is  completely  determined by  the  images of just  three points of the sphere under  the
 transformation. This means   that   the   images after  remapping   of just  three
 retinal points  are  sufficient  to completely  reconstruct the global post-saccade retinal image,   as  the Gombrich
 Etcetera Principle states.\par
Note also  that  if   the  angle $\alpha$ of  the rotation $R^\a_{\underline v}$ is  small  (as it occurs  for  the
microsaccades and other  fixational eyes movements),  then  the   Lorentz   transformation  $L = L_{\Pi, \Pi'}$  is  close
to  the rotation     $R_{\underline v}^{- \alpha}$.\par
\smallskip
\subsection{Consequences for  the Alhazen Visual Stability  Problem}\label{6.4}
The   {\it visual  stability  problem}    consists in finding  an   explanation    of  how  the brain  perceives   the  stable
objects   as  ``stable''  in  spite
of  the changes  of their   retina  images,   caused  by  the  eye  movements.  This  problem     was  first
formulated   in  the eleventh   century  by the  Persian scholar Abu 'Ali   al-Hasan ibn al-Hasan ibn   al-Haytham
(latinised, Alhazen)
and, since then,  it was   discussed  by several  scientists, as  R. Descartes, H. von Helmholtz, E. Mach,  C. Sherrington and
many others (see e.g. \cite{W,W-J-B}).\par
\smallskip
According to the  above conjecture,  the pre-saccade remapping   acts on  the retina  images  as a conformal
transformation.  This    establishes a relation between    the   visual stability  problem for  contours   with    the classical mathematical
problem of  Conformal Geometry   of  {\it   characterising   the   curves of  the  conformal  sphere  up  to
conformal transformations}  (that is,  the conformal
 version   of  the Frenet  Theory of  curves of the Euclidean space).  This mathematical problem   is   completely solved  by means of 
  several results by   A. Fialkov, R. Sulanke, C. Sharp, A. Shelechov    and many others   (see e.g. \cite{F} and references therein). V. Lychagin
 and N. Konovenko   recently gave  a very general  and   elegant  solution to  this   problem   in terms of
 differential  invariants in   \cite{L-K}.\par
\medskip
\section{Differential geometric  models of the  V1 cortex}\label{section 7}
   \subsection{Geometry of  the projectivised  tangent   and   cotangent   bundles of
   $2$-di\-men\-sion\-al  manifolds}
   \subsubsection{The projectivised  tangent   and   cotangent   bundles  of an $n$-manifold}
   Let $M$  be   a manifold of dimension $n \geq 2$  and $\tau : TM \to M$ (resp.,   $\tau' : T^*M \to M$) its  tangent (resp.
   cotangent) bundle. Changing   the  fibers $T_zM$ (resp. $T^*_zM$)  into  their projectivisation  $PT_zM$ (resp.
   $PT^*_zM$), one  gets the  projectivised tangent  bundle  $ \pi: PTM \to M$ (resp.  cotangent bundle   $\pi' : PT^*M \to
   M$)  with  fibers  isomorphic to the real projective  space  $\bR P^{n-1}$.\par
   \smallskip
    Let $(x^1, \ldots, x^n)$ be local coordinates   on  $M$,  and denote by $\partial_i = \frac{\partial}{\partial x^i}$   and   $dx^i$ the
    corresponding coordinate
    vector  fields   and    the  dual coordinate  1-forms, respectively.  The  associated local  coordinates  for the tangent
    bundle  $TM$   are $(x^1, \cdots , x^n, v^1, \cdots v^n)$  where we denote by  $v^i$  the components of  the vectors $ v = \sum v^i
    \partial_i
    \in T_xM$ in coordinate frames. The  coordinates  $(v^1, \cdots, v^n)$ may be  considered  as homogeneous  coordinates  for the projective
    spaces  $PT_xM$.  In  the open subsets of $P TM$  where  $v^n \neq 0$, the associated non-homogeneous  coordinates  are
    $$u^1 = \frac{v^1}{v^n},\qquad \ldots\qquad  u^{n-1}= \frac{v^{n-1}}{v^n}.$$
    This   gives   local coordinates  $(x^1, \cdots, x^n, u^1, \cdots , u^{n-1})$  for $PTM$.
     The  projectivised  tangent bundle   $PTM$ is naturally  identified  with   the   space of the 1-jets of  the
     non-para\-meterised  curves  of $M$ and, hence, it may be  also considered as the   space  of  the (first order)
     infinitesimal curves (or  contours) in $M$.      Any   section
      $$ s :M \longrightarrow PTM\ , \qquad  x\overset{s}  \longmapsto [v_x] = \bR v_x \in PT_xM\ ,$$
  defines  a first order ODE in the manifold $M$. The solutions  of  such  equation   are the non-parametrised  curves
  $\gamma$ in $M$  whose  tangent  lines  $[\dot{\gamma}_x]$, $x \in \gamma$,   coincide with the lines $[v_x]$. In  local
  non-homogeneous coordinates  $(x^i, u^j = \frac{v^j}{v^1})$,  the  ODE associated with a section $s: M \to TM$ reduces  to the
      system
    \beq \label{ODE} \frac{d x^j(x^1)}{dx^1} = u^j(x^1, x^2, \cdots, x^n),\qquad  j = 2, \cdots, n\ .    \eeq
    
     Similarly, a set of  local coordinates $(x^i)$ for $M$  defines an associated system of   local  coordinates $(x^i,
     p_j)$  for  the  cotangent bundle  $T^*M$,   where the $p_j$ are the components of the $1$-forms  $p = p_1 dx^1 +
     \cdots + p_n dx^n \in T^*M$.
    In the coordinate  domains  where  $p_n \neq 0$,
     the projectivised  cotangent  bundle  can be  identified  with the  hypersurface  $H = \{ p_n =1\} \subset T^*M$  with
     coordinates
     $$(x^1, \ldots,  x^n, p_1, \ldots, p_{n-1} ).$$
     \par
    The   projectivised  cotangent  bundle  $PT^*M$   is equipped with   the  {\it canonical contact   structure}  (\footnote{We recall that  a {\it contact structure} on a manifold  is  maximally
    non integrable   field of   hyperplanes in the tangent spaces of the manifold.})
    $\cD \subset T(PT^*M)$,   given by the  projectivisation of  the  kernel  distribution  
    $$\ker \lambda \subset
    T(T^*M)$$
      of  the   {\it Liouville  1-form} of the cotangent bundle $\tau': T^*(T^*M) \to T^* M$, i.e  of the $1$-form on $T^*M$ defined by 
    $$\lambda : T^*M \longrightarrow T^*(T^* M)\ ,\qquad  p \longmapsto \lambda_p:=  (\tau')^*p\ .$$
         In   terms of   the  coordinates $(x^i, p_j)$ for  $T^*M$,  the  Liouville $1$-form $\l$   is the map
         $$p = p_i d x^i\longmapsto \lambda_p =  p_i dx^i\ .$$
         \par
We remark that the    restriction $\lambda_H = p_1dx^1 + \cdots + p_{n-1} dx^{n-1} + dx^n$
       of the  Liouville  form to the hypersurface $H = \{ p_n =1\} \subset T^* M$ is  a  contact $1$-form on $H$  and  defines the  contact  distribution  $\cD_H =
       \ker \lambda_H  \subset TH $,  spanned by the vector fields
        \beq \label{222}  \p_{p_j}, \qquad    \p_{x_j} - p_j \p_{x_n},\qquad  j=1, \ldots, n-1\ . \eeq\par
        \smallskip
         \subsubsection{Special properties   of   $PT M$ and $PT^*M$  when   $\dim M = 2$}
      Let us now   assume   that   $M$ is  a surface, i.e.   $\dim M  = 2$.  In this case,  the projectivised  tangent  and
      cotangent bundles have a canonical identification
        $$PTM = P T^*M$$
        determined as follows:  {\it any line $[v] =  \bR v   \in PTM$, generated  by a vector  $0 \neq v \in
        T_xM$,   can be naturally identified  with  the  line
    $[p_v] \in  PT_x^*M$, given   by the  annihilator $0 \neq p_v\in  T^*_xM$  of  $\bR v$. Vice versa,
    any line  $[p] \in P T^*_xM$  can be   identified with  the line    $[v_p]\in  T_xM$ if  the vectors   $0 \neq v_p \in \ker p$}.
    \par
      \smallskip
    Let  $(x,y)$ be  local coordinates  on $M$ and  $ (x,y, p_1,p_2)$ the  associated
      coordinates  for     covectors $p = p_1 dx + p_2 dy$. As we observed above,  the  open subset      $\{ p_2 \neq 0\}  \subset PT^*M$  can be naturally  identified with the
       hypersurface $H:=\{p_2 =1\}$ of $T^*M$.  This hypersurface has  coordinates  $(x,y, p =-p_1)$  and  consists  of  the covectors of the form 
        $\eta= dy - p dx \in T^*_{(x,y)}M$. The naturally corresponding
       elements of $P TM$  are the lines  of  (annihilating) vectors
       $[v_\eta] = [\p_x + p \p_y] \in  PTM$.
 Reducing \eqref{222}  and \eqref{ODE} to the case $n = 2$, we get that  the contact   distribution
$\cD_H = \ker \lambda|_H $
 is   spanned  at each point  by the pair of  vector  fields 
             $\p_p$ and  $\p_{x} + p \p_{y}$.
        and that  a  section
           \beq \label{fieldorientations} s : M \to PT^*M = PTM, \qquad  (x,y) \longmapsto [v_{(x,y)}]=[ v^1 \p_x + v^2 \p_y],          \eeq
      is associated with  a first  order ODE, which,   in a   domain  where  $v^1 \neq 0$,  has the form
           \beq \label{444}    \frac{ dy}{dx}\bigg\vert_x  =  p(x,y)\ ,\qquad \text{where}\qquad  p(x,y) \= \frac{v^2(x,y)}{v^1(x,y)}  =  \tan \theta_{(x, y)} \ ,\eeq
      where       $\theta_{(x,y)} \in \left[-\frac{\pi}{2}, \frac{\pi}{2}\right)$   is     the {\it orientation} of the covector $p(x,y)$: It is  the  angle between  the  coordinate
     direction  $\p_x$ and  the  line  $\bR\, v(x,y)$, determined by   any  conformally flat metric
      $g = \lambda(x,y)(dx^2 + dy^2)$. Due to this, the sections \eqref{fieldorientations} are called {\it fields of orientations}.\par
\par
\smallskip
       We now recall that any (smooth)   function $F$ on $M$ determines  the section  $dF: M \to T^*M$
       defined  by
       $$dF (x,y) \=\left( \frac{\p F}{\p x} dx + \frac{\p F}{\p y} dy\right)\bigg\vert_{(x,y)} \ .$$
The    complement $M' \subset M \setminus\{ (x,y)\ :\ dF(x,y) = 0\} $ of the critical point set is an open subset,    on which the section $dF: M' \to
 T^*M'$  determines     a  section 
 $$ [dF]: M'  \to PT^*M' = PTM'$$
 of the projectivised cotangent bundle.   Under the above described identification between  $PT^*M'$
and the hypersurface $ H = \{p_2 = 1\} \subset T^* M'$,  such a section $[dF]$ has the form
          $$[dF] =  dy -p(x,y)dx  \qquad \text{with}\qquad p(x,y) = - \frac{\frac{\p F}{\p x}}{\frac{\p F}{\p y}}\ .$$
The   associated    field  of  orientations  $[\p_x + p(x,y)\p_y]$ in the projectivised  tangent bundle $P T M'$  defines the  ODE
        \begin{equation} \label{lagrang} \frac{dy}{dx}= p(x,y)\ ,\qquad p(x,y) = - \frac{\frac{\p F}{\p x}}{\frac{\p F}{\p y}} \end{equation}
Since the vector  field $Z = \p_x + p(x,y)\p_y$ is such that $Z(F(x,y)) \equiv 0$, the  images of the integral curves of the ODE \eqref{lagrang}  are the 
level sets of  the   function $F(x,y)$,  i.e. they are the  {\it contours} determined by  $F$.\par
\smallskip
    \subsubsection{The relation between the projectivize bundle and the circle bundle    of a
  surface}
    Let $g$ be a Riemannian metric on a  surface $M$     and denote by 
    $S(M) = S(M,g) \to  M$  the corresponding   $S^1$-bundle,  given by the unit   vectors 
    $$S(M)= \{\ v \in TM\ ,\  g(v, v)=1\  \} \subset TM \ , $$
    or,
    equivalently,  by the $S^1$-bundle  of the unit  covectors $\xi = g(v, {\cdot})$
    $$S(M)= \{\ \xi \in T^*M\ ,\  g^{-1}(\xi, \xi)=1\  \} \subset T^*M\ .$$
  The  total space $S(M)$ of this bundle can be identified with the  space  of the    infinitesimal  {\it oriented}
  contours and, as the projectivised tangent bundle $P TM$,   is naturally  equipped with    a     contact   structure,  i.e. by the kernel distribution of the  restriction $ \l\vert_{S(M)}$ to $S(M)$
 of   the Liouville  form  $\lambda$  of $T^*M$.\par
   The   total space of the $S^1$-bundle  $S(M) \to M$    is  a  $\bZ_2$-covering   of  the total space of    $PT^*M \to
   M$,   with   projection $\wt \pi: S(M) \to PT^*M$  given by the map that   identifies    opposite  (co)vectors. Local
   coordinates $(x,y)$ of $M$  define   local coordinates $(x,y, \wh\theta)$ for $S(M)$,
where for any vector $v \in S_x(M) \subset T_x M$ we denote by  $\wh \theta \in [-\pi, \pi) $   the  angle  between  $v$ and the  positive  coordinate axis $0 x$. \par
 Note that any  simple smooth   curve $\gamma $ in $M$
 has two natural orientations,  i.e   two  different  ways of moving along the curve.  Each orientation  determines  a
 canonical natural parametrisation $\gamma(t)$ by  the arc length  $t$  and a corresponding  canonical  lift  to  $S(M)$\par
 $$ \gamma(t) \longmapsto  \dot \gamma(t) =  (\gamma(t), \wh{\theta}(t)) .$$
\smallskip
\subsubsection{On the notion   of ``orientation''}   The standard fiber of both   bundles  $   S(M) \to M$  and    $PTM \to M$  is $S^1 \simeq \bP R^1$ and the two  bundles are each other locally isomorphic  by means of the   
local diffeomorphism
$$  \left(x,y, \wh \theta\right) \to  \left(x,y,  \theta= \frac{\wh \theta} 2 \right)\ .   $$
In  many papers  (see e.g. \cite{B-C,S-C-P}),   the   angle  coordinate $\wh \theta \in
[-\pi, \pi)$  of $S(M)$ is called   ``orientation'', but it  is of course different from the  above defined orientation $\theta  \in
\left[-\frac{\pi}{2}, \frac{\pi}{2} \right)$ of the  covectors in  $P TM$.  We do not follow this  habit and, 
through the rest of  the paper,   we  will  use  the name ``orientation''  just  for  the  above defined coordinate   
$\theta$  of $PTM$.\par
       \smallskip
       \subsection{Hoffman's  pioneering model   of  the   V1 cortex} \label{hoffman}
 The  first  attempt  to    develop  a  differential geometric model for  the  primary visual cortex  V1    appeared  in a pioneering  and  very stimulating  paper  by W.  Hoffman  \cite{Hof}. In particular, among other  important ideas,  for the first time Hoffmann pointed out  the crucial  role of the M\"obius conformal group  $\SO(1,3)$  in descriptions of  the functional structure of the primary visual cortex. 
 \par
  Hoffman  proposed a model,  in which       the
 V1  cortex  is  mathematically represented   as  the  $3$-dimensional  total space  $V = \mathcal{C} R$    of a  fiber bundle
  $$  \pi : V = \mathcal{C} R \to  R$$
  over  the  retina  $R$,   whose   fibers  correspond to   the columns  of   the  V1 cortex.  According to this model, the simple neurons of   a column  form an
   {\it orientation response field}  (ORF) and    the collection of all such ORFs 
   determine  lifts of  retinal  contours to  curves in  $V$.  The  tangent  vectors  to  these lifts   define
     a contact structure     on  $V$. 
In Hoffman's words, 
\begin{itemize}[leftmargin = 5pt]
\item[] ``{\it the thing one first thinks of is that the visual contours are integral curves
of the cortical vector field embodied in the ORFs. In other words, the visual
map is a tangent bundle  $  T\mathcal{C}R = V  \to   R$, where    $ R$ denotes the retinal
manifold and V the cortical “manifold of perceptual consciousness.” But the
ORFs unfortunately do line up head to tail as in an Euler line approximation
to an integral curve. Furthermore, the ORFs have an areal character as well
as a line-element one.}''
\end{itemize}
\par
\smallskip
Unfortunately,   in   \cite{H}  diverse mathematical errors and inaccuracies occurred. For instance, it was   stated that   the  V1 cortex  should be identified with    the tangent bundle  $T { R}$ or   cotangent  bundle $T^* { R}$ of the retina,  while on the contrary it is now known that it should be identified with  the projectivisations of these bundles.  It   is also claimed  that the  conformal group $\SO(1,3)$  acts transitively on the  contact bundle, but this is not true. 
Other   faults of this kind  appeared.  \par
However, despite of these problems, there is no doubt that Hoffmann's   ideas promoted and strongly  influenced  all  subsequent developments  of differential  geometric models for  the visual system.   Following a suggestion of J. Petitot, the area of applied geometry which originated from Hoffman's work  is nowadays  called {\it neurogeometry}.\par
\smallskip
\subsection{Petitot's  contact models}\label{subsection  7.2}
After  Hoffman's  pioneering model,    in a paper by J. Petitot and   Y. Tondut   and in several subsequent works of J. Petitot 
 a  new precise  contact  model   for the  V1 cortex  and  laid the foundations of neurogeometry was  developed  systematically and in great detail. \par
  \smallskip
 More precisely,  in \cite{P-T, P} (see also \cite{P-T, S-C-M,C-S0,S1,B-Z})  it was introduced a model for the V1 cortex, in  which 
   the  retina  $R$   is identified  the Euclidean plane $R
 = \bR^2$  and  the V1  cortex   is identified with   the projectivised cotangent   bundle
 $$\pi:   PT^* \bR^2 =  \bP R^1  \times \bR^2 \simeq S^1 \times \bR^2  \longrightarrow  R=\bR^2\ .$$
  According  to  Petitot's model, each   fiber  $\pi^{-1}(z)$, $z \in \bR^2$,  of  the bundle $P T^*\bR^2$   corresponds to  a   pinwheel with RF  $z \in R = \bR^2$, while 
   each point $(z, \theta)$ of  a  fiber  $\pi^{-1}(z)$ corresponds to a  single simple neuron  of  a  pinwheel. ln this way, the simple neurons of the pinwheels  of the V1 cortex   are parameterised by  their receptive fields $z \in R = \bR^2$ and their  orientations $\theta$.
 In such a model,  given   a  (regularised)  input  function  $I(x,y)$  on retina,  a simple neuron   with receptive field $z_o$ and  
       orientation $\theta_o$   fires   when  there is a  retinal  contour $C  = \{I(x,y)= \text{const.} \}$  that passes through   $z_o = (x_o, y_o)$ and  with tangent line at $z_o$  with 
       orientation  $ \theta_o$.  \par
       \smallskip
       We recall that    a retinal
    contour $C = \{I(x,y)= \text{const.} \}$ is   a non-parametrised  curve. For such a curve we  may  always  choose a   parametrisation,  say $\g(t) =
    (x(t), y(t))$, $t \in (a,b)$.  If we    assume that  $\dot{x}(t) \neq 0$, we may  also  change the parameter from  $t$   to
    $x$ and   represent   the  contour $C$
      as the image of a parameterised  curve of the form  $z(x)=(x, y(x))$.
     Denoting   by  $x \mapsto\theta_{x} \in \left[- \frac{\pi}{2}, \frac{\pi}{2}\right) $  the map which gives  the  orientations of the tangent line of the curve for any $x$, as we pointed  in \eqref{444}, the function   $y(x)$
     is a solution to   the  ODE
     \begin{equation}\label{ODE-1}
     \frac{dy(x)}{dx} = \tan \theta_{x}.
     \end{equation}
     Therefore,  given a contour $C $,   the  curve $c(x) =  (x,y(x), \theta_x)$, given by  the fired   neurons,  is  a curve in  $PT^* \bR^2$ (= the V1 cortex), which is a lift of the curve $(x, y(x)$ and  with   tangent vectors 
         $$\dot c(x) =\p_x + \frac{dy(x)}{dx}\p_y  + \dot \theta_x \p_{\theta} =  \p_x +  \tan \theta_x \p_y  + \dot \theta_x \p_{\theta}\ .$$
   The  curve $ c(x) = (x,y(x), \theta_x)$ is  {\it horizontal}  in the sense that   it is tangent to the contact
        distribution  $\cD_H$  of  $ P T\bR^2  \simeq  H  = \{ v^1 = 1\} \subset T^* \bR^2$ at each point: Indeed, 
       $$ \lambda_H (\dot c(x)) = (dy -  \tan \theta_x dx)(\dot c(x))= \tan \theta_x -   \tan \theta_x   =0\qquad \text{for any}\ x\ . $$
The  curves in  the  cotangent bundles  that are  lifts  of this kind of curves on  the base manifold are  called   {\it Legendrian lifts}.  We may therefore say that in   Petitot's model  {\it the firing of the simple neurons in the pinwheels of the V1 cortex carry out   the Legendrian lifts
of  the contours}. \par
\smallskip
One of the  main  aims of the  processing of the visual system  is  to  integrate the  information, which is encoded in
the  firing of the  neurons,  and obtain  in this way  a  global description of  the  contours, that is of the curves   that are solutions to the
differential equation  (\ref{ODE-1}).\par
  \smallskip
Note that the  most important retinal contours  are  {\it the  closed  contours,   which  constitute   the  boundaries  of  retinal
images  of   three dimensional  objects} (see \cite{Ka,K-I-G-W}).  Any closed contour divides the  retina  into two
parts:  one  of them  is  the  `` image  of an object'', the other  is  the background on which  the  object is located.
The   determination   of which the two parts is  the ``image of an object''   is  mathematically equivalent to  fixing  an orientation of the contour, that is  
to   choosing  one of the  two possible    directions for   a movement  along the  contour: Once a region is considered as  ``image of an object'',  there is a unique orientation for the contour
 such that  the ``image'' is  on its left  with respect to  the orientation and, conversely, once an orientation is fixed, there is only one  region  that  lies on  the left of the contour with respect to the orientation and thus  only one region which is considerable as  the   ``image of an object''.\par 
\smallskip
  It  would be  interesting   to understand  where   and  how the  visual system   determines  the orientations of  the
  contours,  that is   how it    lifts     curves of the {\it non-oriented} infinitesimal contours of   $PT^*\bR^2$  into
  curves  of   {\it oriented} infinitesimal contours    in  $S(\bR^2)$.   Optical illusions
 show  that   when  both  sides  of  a  closed contour   admit  meaningful interpretations,   the visual
 system involuntarily and periodically    changes    the  orientation  of the  contour. This indicates  that  the orientations
 of  contours are  determined  by the   high levels of   the visual system.\par
 \medskip
\subsection{Sarti, Citti and Petitot's symplectic  model}\label{subsection 7.3.}
In  \cite{S-C-P}, the  authors  proposed   a symplectisation of  Petitot's  contact model,  the so called  {\it  symplectic model}. According to such a   model, 
 each    simple  cell  of the V1 cortex  is  characterised  not only by its RF $z$ in the retina  (which, as in previous Petitot's  contact model, is identified with $ \bR^2$) and  its  orientation  $\theta$, but also 
  by a new  parameter $\s$,  which is called
{\it scaling} and corresponds to  the intensity of    the  reply  of the neuron to a   stimulus.
  In more detail, for  the symplectic model, 
 the  V1 cortex is   represented by the  total space of a principal bundle
 over   $R = \bR^2$
$$\pi : P = \bC^* \times R  \longrightarrow R\ ,$$
 with fiber  given by the   group of the non zero  complex numbers  $\bC^* = S^1 \times \bR^+ = \{\s e^{i \theta}  \}$.
Such a  $\bC^*$-bundle $P$ can be  also identified with either of the   following  manifolds: 
\begin{itemize}
\item the group $P = G$ of similarities of $\bR^2$
 $$  G:= \Sim(\bR^2) =\left(\bR^+ {\cdot}\SO(2)\right)   \ltimes T_{\bR^2}  = \bC^* \ltimes \bR^2 ,$$
  where  $T_{\bR^2} = T_{\bC}$ denotes  the group of  the parallel translations of $\bR^2$
  \item the  cotangent  bundle   with the   zero section removed  $T^*_{\sharp} R$ of the retina $R = \bR^2$. 
  \end{itemize}
 By the  identification       $P  =  T^*_{\sharp}R $, the  bundle   $P$  is naturally  equipped  with  a   {\it symplectic structure}, namely the non-degenerate closed
 $2$-form
  $\omega = d \lambda$ determined by   the Liouville form $\lambda$ of $ T^*_{\sharp}R$.
In what follows,   we  call    the symplectic bundle $\pi: P  = G =  T_{\sharp}^* R \to R$   the {\it Sarti-Citti-Petitot (SCP)
    bundle}.\par
        \medskip
 The identification between the V1 cortex and the bundle $P = \Sim(\bR^2) \to R = \bR^2$,  given in   Sarti, Citti and Petitot's  symplectc model,  is based on  the one-to-one correspondence between  neurons and points of  $P =  \Sim(\bR^2)$ defined as follows.   Consider the
complex mother Gabor filters $\Gabor_{\g^\bC_0}$, with the  RP    \eqref{mothergaborcomplex} and let  $\left(\Gabor_{\g^+_0},
\Gabor_{\g^-_0}\right)$  be the  real  mother  Gabor filters, whose  RPs are   by the real and imaginary parts of $\g^\bC_0$. Let also  $(n^+_0, n^-_0)$ be a  pair of  simple  neurons of the V1 cortex that work as  the pair of mother  Gabor filters $\left(\Gabor_{\g^+_0},
\Gabor_{\g^-_0}\right)$.  As we explained in \S  \ref{sect3.5},
 any element  $ a {\cdot}  T_{b} \in \bC^* {\cdot} \bC $,   transforms
  the pair of  $(\g_0^+, \g_0^-)$ into  the  pair $(\g_{a, b} ^+, \g^-_{a,b})$ defined in \eqref{tranG}. The symplectic
  model of Sarti, Citti and Petitot  is based on the  following fundamental assumption: \\[5pt]
   {\it  the V1 cortex  consists of pairs of simple neurons  $(n^+_{a,b}, n^-_{a,b})$, in which the neurons  $n^\pm _{a,b}$ work  as  the   Gabor filters
  $(\Gabor_{\g_{a,b}^+}, \Gabor_{\g_{a,b}^-})$   and are  therefore parameterised  by  the elements  the similarity group  $
  P =\Sim(\bR^2) = \bC^* \times \bC$. }
\medskip
 \subsection{A physiological interpretation for the coordinates $(z, \theta, \s)$ in  Sarti, Citti and Petitot's  model: the  Principle of Maximal Selectivity}\label{subsection 7.4}
 In \cite{S-C-P}   (see also  \cite{C-S0}, p. 312),  Sarti, Citti and Petitot  consider the following {\it Principle  of Maximal
 Selectivity}:\\[5 pt]
``{\it The intracortical circuitry is able to filter out all the spurious directions and to strictly keep the direction of
maximum  response of the simple cells}".\\[5 pt]
According to this  principle, for any  input function $I: R \to \bR$, at each point $z = (x, y) \in R$, there is 
  a  neuron with RF $z$   which provides  the maximal output for the stimulus.  According to the symplectic model, 
such a neuron corresponds to a point of the fiber $P|_z = \pi^{-1}(z) \simeq \bC^*$ and it  therefore corresponds to 
 special values    $(\overline \theta_z,  \overline \sigma_z)$ for  the coordinates of the fiber $\bC^* = S^1 \times \bR^+$. 
As a consequence of the Principle of Maximal Selectivity, a stimulus  determines   a section  of the SCP bundle $\pi: P\to R$:
$$\r: R \to P\ ,\qquad \r(z) \= (z, \overline \theta_z, \overline \s_z)$$
given by the neurons with maximal responses. 
 Since $P$ is identified with the cotangent bundle $P = T^*_{\sharp} R$, this  section
 can be  considered  as a 1-form  $\r$ on $R$.  Sarti, Citti and Petitot  proved that  $\r$ is closed and hence that the
 surface   $\r(R)\subset T^*_{\sharp}R$ is Lagrangian.\par
 \smallskip
 By the results in \cite{S-C-P}, given a point  $z$ and an input function $I$, the above defined $1$-form $ \r(z) \= (z,
 \overline \theta_z, \overline \s_z)$  is such that  the real number $\frac{1}{\sqrt{2}} e^{\overline \sigma_{z}}$ is
 essentially equal to the distance of $z$ from the nearest  point of a nearby boundary of a figure (that is a contour  on
 which the gradient  $dI$ is very large) and $\overline \theta_z$ is equal to the orientation of the tangent line of this
 boundary at such  nearest point in the boundary.   This immediately  leads to the following physiological  interpretation for the parameters that distinguish   simple neurons with the same RF: \\[5 pt]
 {\it a neuron with RF $z$ and  corresponding to the pair of parameters $(\theta, \s)$ 
 fires when $\frac{1}{\sqrt{2}} e^{\overline \sigma_{z}}$ is
  equal to the distance of $z$ from the nearest  point of the  boundary of a figure and $\theta$  is the orientation of the tangent line of this
 boundary at such  nearest point}\par
\smallskip
This is a physiological interpretation for the coordinates  $\theta$ and  $\s$  of  the SCP bundle, which is however   based on purely non-local properties of the input function $I: R \to \bR$. In fact, 
according to the such an interpretation, {\it the firing of a simple neuron does not depends on the {\rm infinitesimal} behaviour of $I$ at a RF $z$, but on its behaviour  at {\rm finite} distances from $z$}.
This means that, following Sarti, Citti and Petitot's  discussion, the scaling $\s$ cannot be  considered as an ``internal parameter''  in the sense of   Hubel and Wiesel.\par
\medskip
  \section{Models  of  hypercolumns}\label{section 8}
\subsection{Bressloff and Cowan's spherical model of  a hypercolumn} {\label{subsection  8.1 }
  As we mentioned in the Introduction and in \S \ref{hubell},  according to Hubel and  Wiesel  a hypercolumn is a  collection of  columns of  the V1 cortex,  in
   which  the internal parameters  may take all possible values.
      P. Bressloff   and  J.  Cowan proposed   a {\it Riemannian spherical  model} for the  hypercolumns,
      based on the assumption that  the  internal   parameters are just two: the orientation $\theta$  and  the   spatial  frequency $p$ \cite{B-C,B-C1,B-C2}.\par
      \smallskip They   assumed
      that a  hypercolumn $H$ is a domain in the V1 cortex,     associated  with two  pinwheels $\Sp,\Np$,       corresponding
      to the minimal  and   maximal values  $p_-$,  $p_+$ of the spatial frequency.  According to  the   model, the
      simple  neurons  of the hypercolumn $H$ are  parametrised  by  the orientation $\theta$  and  the {\it
      normalised  spatial frequency} $\phi$,  defined by
       $$\phi \=\pi \frac{\log(p/p_-)}{\log(p_+/p_-)}  \ .$$
 The normalisation is chosen   in such a way that  $\phi$   may assume only values  in the interval
$[0.\pi]$.
For any choice of  $\theta$ and $\phi$, there is a corresponding  simple  neuron  $ n = n(\theta, \phi)$  in $H$, which fires  only  if  a
stimulus occurs  in its RF and   with     orientation normalised  spatial  frequency   $\theta$ and $\phi$. \par
     \medskip
   Bressloff and Cowan proposed  a mathematical  interpretation of  these parameters   as  spherical coordinates of the $2$-sphere $S^2$,  with
   $\theta \in [0, 2\pi)$ corresponding to  the {\it longitude} and $\phi = \phi' + \pi/2 \in [0, \pi]$  to  the  {\it
   polar  angle} or  {\it shifted latitude}  (with  $\phi' = \phi- \pi/2$    {\it latitude} in the usual sense).      The
   shifted  latitudes of the pinwheels  $\Sp,\Np$  are $0, \pi$,  but the   longitude (=orientation) is not defined for them
   -- in fact, the pinwheels  are able to detect contours of {\it any}  orientation.  Following this mathematical interpretation of $\theta$ and $\phi$, Bressloff and Cowan identified an
   hypercolumn $H$  with  the sphere $H_{BC} =S^2$ and the 
   pinwheels $\Sp, \Np$ of $H$   with   the south  and the north pole  of  $H_{BC} = S^2$. Then they 
 used such a clever  model to describe the evolution of the excitation of
visual neurons according to the Wilson-Cowan equation and got many  interesting results.\par
\smallskip
We now recall that the RF  of a   simple neuron  $n$ is very small, so  that it can be  considered  as a single point $z(n)$ of the  retina. In Bressloff and Cowan's spherical model, 
  the map  that sends   each simple   neuron $n$ of  a hypercolumn  $H$ to  its  RF  $z(n)$  is  considered  as  a smooth map between the  surfaces
$$ z :  H_{BC} =S^2 \longrightarrow  R_H :=  z(H_{BC})\ , \qquad   n \longmapsto  z(n) $$
where we denote by  $R_H$ the the region of the retina, which is the receptive field of the whole hypercolumn (i.e. the union of the RF of all its neurons).
This map (or, more precisely, its restriction to some appropriate open set) is assumed to be a diffeomorphism. 
   This implies  that, {\it in Bressloff and Cowan's model, the spherical coordinates 
   $(\theta, \varphi) $ of  $ H_{BC} $ may be  considered  also  as  coordinates for (an open subset of)  the  receptive field    $R_H$ of the hypercolumn $H$ and  cannot   be considered as  internal
   parameters in the sense of Hubel and Wiesel. } Only   for  the neurons in the   pinwheels $\Sp, \Np$   of the hypercolumn  (i.e. those on   which  the
   ``latitude'' $\varphi$ takes extremal values)  the parameter  $\theta$ can be considered as an  ``internal parameter''.
   \par
\medskip

\subsection{The need of improvements for   Sarti, Citti and Petitot's symplectic model and  Bresloff and Cowan's spherical model}
In this short section, we would like to point out  a couple  problems for 
  Sarti, Citti and Petitot's symplectic model and for Bresloff and Cowan's spherical model, which indicate  
the need of some improvements for  those two   models.  A development in this direction is  given by  the   {\it conformal spherical model of a hypercolumn} and the corresponding {\it reduced model},  discussed  in the next subsection.\par
\medskip  
In their  symplectic model, Sarti, Citti and Petitot  describe  the  simple neurons of the V1 cortex   in terms of the   bundle   of
conformal frames
$$\pi:  \Sim(\bR^2)  \longrightarrow     \bR^2 = \Sim(\bR^2)/ \CO(2)$$
over the    retina   $R = \bR^2$.    In this model,  the simple neurons of the V1 cortex appear in   pairs $(n^+_{a,b}, n^-_{a,b})$ and  the neurons  $n^\pm _{a,b}$ of  any such pair  work  as  the  Gabor filter
  $\Gabor_{\g_{a,b}^\pm}$  defined in \S \ref{sect3.5}. They  are  therefore parameterised  by  the elements  $(a, b)$ of the similarity group  $
  \Sim(\bR^2) = \bC^* \times \bC$. Each   fiber  $\pi^{-1}(z)  \simeq  \bC^*$     represents a  column     with the  RF   $z$ and the (pairs of) simple   neurons   
in  such a   column  are  parameterised   by
the  two coordinates of $\bC^* = S^1 \times \bR^+$, the orientation  $\theta$  and    the scaling    $\sigma$.   In particular, according to this model, {\it  the columns  of the V1 cortex represented by the fibers $\pi^{-1}(z)$, $z \in R$, of the SCP bundle  are just the {\rm  pinwheels}:  In fact,   they contain simple neurons that are able to detect any possible orientation $\theta$ for the contours that pass through their receptive fields $z$}.  \par
\smallskip
This  remark shows that, in  the symplectic model,  the V1 cortex  behaves  as a collection of pinwheels with no regular column among them. Since it is known that most of the columns of the V1 cortex are regular,  we think that  this represents    a somehow weak point of the symplectic  model. 
\par
Let us now focus on  Bressloff and Cowan's spherical model. In this model the simple neurons  of a  hypercolumns  $H$ are mathematically represented as points of a sphere $H_{BC} = S^2$ and are distinguished each other by  the two spherical coordinates of $S^2$,   identified with the  orientation
$\theta$ and the   normalised   spatial frequency   $\phi$   of the neurons. But, as we pointed out at the end of  \S \ref{subsection  8.1 }, in Bressloff and Cowan's model, 
the spherical  coordinates $\theta, \phi$   cannot  be considered as internal parameters for the hypercolumn (with the exceptions of the neurons in the two pinwheels $\Sp$, $\Np$ of the hypercolumn). 
In other words, according to  the spherical model,  most of the simple neurons of a hypercolumn are not associated to any value of an internal parameter  in the sense of Hubel and Wiesel. We think that  also this aspect 
  of the spherical
model is a weak point. \par
\medskip
As we will see in the next two  subsections,  the conformal spherical model offers an improvements of Bressloff and Cowan's model, in which new   internal parameters occur and  yields a simplified 
version, which can be considered a  refined  version of Sarti, Citti and Petitot's model of the V1 cortex and suggests an alternative physiological interpretation of the scaling parameter $\s$, based on purely local properties of the input function. 
\par
\medskip
\subsection{The  conformal spherical model  of a hypercolumn and the associated reduced model} \label{subsection  8.2.}
We now    discuss the {\it conformal spherical  model},   shortly announced in  \cite{A1} and presented in     detail in  \cite{A-S}.
\par
\smallskip 
We recall that in   Bressloff and Cowan's Riemannian spherical     model,  a hypercolumn $H$ is mathematically represented    as    a  sphere
    $H_{BC}=S^2$ equipped with  the    standard  round  metric $g_o$.  The  basic new idea introduced in \cite{A} was to  construct an improvement of  such spherical model  replacing  the   standard Riemannian sphere $(S^2, g_o)$  by   the   {\it conformal  sphere}  $(S^2, [g_o])$, i.e. by the $2$-sphere equipped with the conformal structure $[g_o]$ (see \S \ref{section 5} for main definitions and fundamental results). \par
 \smallskip
 Inspired by the  method  of  the construction of    Citti, Sarti and Petitot's symplectic model, the fundamental assumption of the conformal spherical model is that 
 the simple neurons of a  hypercolumn appear in pairs $(n^+_g, n^-_g)$, in which  the neurons $n^\pm_g$  work  as the Gabor filters with 
 RPs $\gamma^\pm_g$ given by the  real and imaginary parts of the   RP 
$$\gamma^{\bC}_ g(x , y ) = \frac{1}{\operatorname{J}(g)(x,y)} \gamma_0^{\bC} (g^{-1}(x, y)) $$
where
\begin{itemize}
\item[--]  $\gamma_0^{\bC} (x, y)$ is  the RP of the   ``mother'' Gabor filter defined in \eqref{mothergaborcomplex}; 
\item[--] $g$ is an element of the connected group $\operatorname{Conf}^o(S^2) \simeq \SO^o(1,3)$ of  conformal transformations of $(S^2, [g_o])$  (see \eqref{8}).
 \end{itemize}
 In other words, the neurons $n^\pm_g$ behave as the Gabor fields, whose  RPs   are the real and imaginary parts of the complex Gabor filter, which is obtained  from the  ``mother'' complex Gabor filter $\Gabor_{\g^\bC_0}$ by changing  the coordinates $(x, y)$ by the conformal transformation $g$ and using the fact  that, under changes of coordinates, the RP profiles transform as densities -- see \eqref{density}.\par
 \medskip
 In this way, according to the conformal spherical model, {\it the (pairs of) simple neurons  of a hypercolumn are parametrized by   the elements of   $\operatorname{Conf}^o(S^2) \simeq \SO^o(1,3)$}.
\par
\smallskip
On the other hand, as we pointed out in \S \ref{Cartanapproach}, the group $ \operatorname{Conf}^o(S^2)$ is the total space of the homogeneous bundle 
\beq \pi:  \operatorname{Conf}^o(S^2) \longrightarrow S^2 =  \operatorname{Conf}^o(S^2)/\Sim(\bR^2)\ ,\label{thebundle} \eeq
 which is in turn identifiable with the total space of the bundle 
$\pi:  \mathcal{CF}^{(2)}(S^2) \longrightarrow S^2$ of the second order conformal frames of $S^2$. This means that, {\it according to the conformal spherical model, a hypercolumn $H$ is identified with the total space
of the bundle \eqref{thebundle} with fiber given by the $4$-dimensional group $\Sim(\bR^2)$}.\par
In this model, the coordinates of the fiber $\Sim(\bR^2)$ are internal parameters for the simple neurons of the hypercolumn $H$. 
{\it If   such   fiber  coordinates are ignored,  one gets that  the neurons of the hypercolumn are parameterised just by the two spherical coordinates of $S^2$ 
and  Bressloff and Cowan's   spherical model is recovered.}\par
\smallskip
On the other hand, the conformal spherical model and  Riemannian spherical model crucially differ in the following aspects: 
\begin{itemize}
\item According to  the Riemannian spherical model, any point,  which is different from  the north pole $\Np$  and south pole $\Sp$ of  Bressloff and Cowan's sphere   $H_{BC} = S^2$ (or, more precisely, each point of an open and dense subset $\cU \subset H_{BC} \setminus \{\Np, \Sp\}$),  is  assumed to be the  
RF of exactly  one  simple neuron, while   the  north pole $\Np$  and the south pole $\Sp$ are assumed to be  the RFs of two pinwheels, the only two columns that Bressloff and Cowan assume  to be  contained in  a  hypercolumn;  
\item  According to  the conformal spherical model,  any point $z \in S^2$ --  with no distinction --   is assumed to be the  common RF (or  the union of  several  RFs, but all of them very close each other) of  a system of   simple neurons,  each of them distinguished from the other  by $4$ internal parameters. Assuming that one of  the internal parameters is the orientation $\theta$, the system  of  simple neurons with a given  RF $z$ is  organised into a collection  of  several  sub-systems,  one per each  value of the internal parameter $\theta$,  and containing  simple cells that are parameterised by the remaining $3$ internal parameters. 
\end{itemize}
 In the conformal spherical model,  the north and south poles  $\Np, \Sp \in S^2$ are still the common RFs of  many neurons, possibly organised in  regular columns or  pinwheels (as it occurs in the Riemannian spherical model), but  -- in contrast with the Riemannian spherical model -- the  poles $\Np$ and $\Sp$ are no longer the only points of the sphere with such a property. In fact, according to  the conformal spherical model {\it all points of $S^2$ are assumed to the  common RFs  of a large family of simple neurons}.\par
 \smallskip
As we mentioned above,   the  conformal  spherical  model  naturally leads to a second  model for    small neighbourhoods of  the poles of $S^2$, called {\it reduced model}, 
  which can be considered as    a  local version  of   Sarti, Citti and Petitot's symplectic model. Such a reduced model  offers  a   physiological interpretation  for   the  scaling parameter $\s$   in terms of   the  normalised  logarithm of the  spatial frequency, which is a purely local property of the input function.\par
  \smallskip
As we showed in \S \ref{Gaussdecomp}, there is an open and dense subset $\cU \subset  \SL(2, \bC)$ which admits  a Gauss decomposition of the form 
$\cU = N^- {\cdot} \bC^* {\cdot} N^+$, where $N^\pm$ and $\bC^*$ are the subgroups of $\SL(2 \bC)$ defined in \eqref{gd}.
   On the other hand,    according  to Lemma \ref{Lemma} and the subsequent remark,  the  $1$-dimensional complex subgroup
   $N^+ $   acts  on  the  tangent  space $T_{\Sp} S^2$ 
    as  the  group of  parallel translations while the local actions  of  the elements  in $N^{-}$  are very close to the identity map near $\Sp$.\par
    \smallskip
   Let us now pick a  small neighbourhood $\cV \subset T_{\Sp} S^2$ of the origin  of   the tangent  space  $T_{\Sp} S^2$   and denote by 
     $\cV_{\Sp} \subset  S^2$  the     neighbourhood of the south pole  $\Sp$ in $S^2$,  which corresponds to $\cV_{\cS}$ via the stereographic projection $\ster_{\bN}: S^2 \setminus\{\Np\} \to T_{\Sp} S^2$. Since the map $\ster_{\bN}$ is  $N^+$  equivariant,   the  open   subset  $\cU^{(N^+)} \subset  N^+$  of  the  elements  in $N^+$ that transform the origin $0 \in T_{\Sp} S^2$  into the other points of  $\cV \subset T_S S^2$, is such that 
     $$ \cU^{(N^+)} \cdot \Sp =  \cV_{\Sp}\ .$$
   Ignoring    the      closed to identity  local action    of the   transformations in  $N^-$ (i.e.  focusing just on the group  $ \bC^* {\cdot} N^+$ and neglecting the contribution of $N^+$
    $  \cU = N^-{\cdot} \bC^* {\cdot} N^+$  of $\SL(2 \bC)$, ), we get 
    that the  simple neurons  corresponding to the fibers over  $\cV_\Sp \subset S^2$ are represented  just by    the elements of  $\bC^* \times  \cU^{(N_+)} \subset \bC^* {\cdot} N^+ $. 
 Such simple neurons  occur in  pairs    $(n^{+}(g), n^-(g))$,     associated with 
  the  element $g \in \bC^* {\cdot} \cU^{(N^+)}$,       and work as pairs of Gabor filters with   RPs, which are the real and imaginary parts  of the complex  Gabor   filters, with  RP are  
     obtained  from the mother complex   RP $\gamma_0^{\bC}$  by the    transformations  $g \in\bC^* \times  \cU^{(N_+)} $.\par
     \smallskip
     This remark lead to the {\it reduced model}    which consists precisely   of  the  following assumption:  {\it  the simple neurons of a hypercolumn  corresponding to the points of a (sufficiently small) neighbourhood $\cV_{\Sp} \subset S^2$  of the south pole $\Sp$ of  $H_{BC} = S^2$  are parameterised by an open subset of the form $\bC^* {\cdot}  \cU^{(N+)} \subset \bC^* {\cdot} N^+$  of the group $\SL(2,\bC)$.}
\par
\bigskip
Notice that  the  reduced model  can be considered as  a   localisation of the  Sarti, Citti and Petitot's symplectic model, in which the
translation  group   $ \bR^2 \subset \Sim(\bR^2)$ is replaced by a  sufficiently small open subset  $\cV_\Sp \subset N^+$ of $N^+ \subset \SL(2, \bC) $.  
The  replacement of the non-compact manifold  $\bR^2$ with a small  subset $ \cU^{(N_+)} $ of $N^+ \simeq \bR^2$ make the symplectic model  more realistic, even if     there is still an unrealistic assumption  that the neurons 
corresponding to  a  fiber  $\pi^{-1}(z) \simeq \bC^*$  are parametrised  by the points of  the non-compact
    group  $\bC^*$.   In order to reach an even  more realistic model,   it would be  natural    to   ``localise''  not only the base manifold but also the     fibers  of the bundle (i.e. 
 to    replace     the standard   fiber    $\bC^*$  by    a compact  neighbourhood of  the identity). An improvement in this direction is
discussed in   \cite{A-S}, to which we refer for the details.\par
\medskip

According to the  reduced model,   the neurons associated with the points an appropriate portion $\cV_\Sp$ of the Bresloff and Cowan's  sphere $H_{BC} = S^2$
are  parameterised by the   coordinates for the $4$-dimensional  space $ \bC^* \cdot \cU^{(N+)} \subset \bC^* {\cdot} N^+$:   two real parameters, given by modulus and the argument  of a complex coordinate of  
the  open subset   $ \cU^{(N^+)} $ of the group  $  N^+ (\simeq \bC)$  and two  real parameters, given by the modulus  and the argument   of a complex coordinate   
    for the group $\bC^*$. In \cite{A-S} it is pointed out that the two real  parameters for the subset  $\cU^{(N^+)}\subset N^+$, say $\wt \theta$ and $\wt \f$,  can be used in two different ways:
    \begin{itemize}
    \item[(a)] to parameterise the projections of the points onto the base of the bundle $\pi:  \bC^* \cdot \cU^{(N^+)} \to \cV_S \subset S^2$ or
    \item[(b)]  to parameterise the points of the fiber of the bundle $\pi:  \bC^* \cdot \cU^{(N^+)} \to \cV_S \subset S^2$. 
    \end{itemize}
    If   $\wt \theta$ and $\wt \f$ are used as in (a), they are 
``external parameters'' for the considered collection of  simple neurons, i.e.  they parameterise the RFs in the sphere $H_{BC} = S^2$ and they can be identified with  the {\it orientation}   and the  {\it normalised logarithm of  the  spatial
frequency} of the neuron, as in Bressloff and Cowan's Riemannian spherical model. On the other hand,  if $\wt \theta$ and $\wt \f$ are used as in (b), they are ``internal parameters'' for the fiber bundle and they can be identified with
the internal parameters  $\theta$ and $\s$ of   Sarti, Citti and Petitot's  symplectic model.  A comparison between  these  two interpretations for  the parameters $\wt \theta$ and $\wt \f$ immediately yields   to {\it a  physiologically interpretation of the  scaling parameter $\sigma$ as the Bressloff and Cowan's normalised logarithm of  the  spatial
frequency of the neurons}.
\par
\smallskip
In \cite{A-S} it is also  remarked  that  an analog of the above  reduced model  may be  constructed also for the neurons corresponding to the points of a    neighbourhood  $\cV_{\Np}$  of the north pole or of any   other   point  $z \in H_{BC} = S^2$. It is also noticed that 
the neurons with RF given by the north pole $\Np$ (resp. south pole $\Sp$) are  characterised by  having the maximal (resp. minimal)  value  for  the   spatial frequency. It is  known  that
   the   processing     of high-frequency and low-frequency information differs from each other  \cite{E-P-G-K-S-K}  .\par
   \smallskip
These observations   suggest the possibility of the  following  new speculative  model  for a hypercolumn.  As in Bresloff and Cowan's model,  assume that the 
RF of the  neurons of a hypercolumn $H$ constitute  a retinal domain $R_H \subset R$, which is diffeomorphic to  (an open and dense region of) Bressloff and Cowan's sphere $H_{BC} = S^2$ 
and are therefore parameterised by  the  spherical coordinates   $ \theta, \phi$ with the above described interpretations in terms of orientation and spatial frequency. 
Let us  now split   the  retinal     domain   $R_H$ into  two   subdomain   $R_H  = \cV_\Sp \cup \cV_\Np$, corresponding to    the   southern   hemisphere (centred at the south pole  $\Sp$) and the  northern  hemisphere (centred at the north pole $\Np$).     Then let us   consider
{\it two    copies}
 of the   reduces  conformal model,   one   for the system  of neurons  associated with $\cV_\Sp$   and  another  for the system of neurons   associated   with     $\cV_\Np$.
 We conjecture  that the   first  (which identifies an appropriate  system of neurons with a bundle  of the form
$   \pi_\Sp :  P_\Sp = \bC^* {\cdot} \cU^{(N^+)}  \to  \cV_\Sp$ for some   $\cU^{(N^+)} \subset  N^+$) can be used as model for  a system of neurons 
 that processes  low spatial frequency visual information, while the second (concerning a bundle  $   \pi_\Np :  P_\Np = \bC^* {\cdot} \cU^{(N^-)}  \to  \cV_\Np$ for some   $\cU^{(N^-)} \subset  N^-$)
can be associated with a system of neurons 
 that processes  high spatial frequency visual information.
\par
\bigskip

\subsection{The  Principle of Invariance and an  application of   the  reduced  model}\label{subsection
8.4.}

   Let  $K$  be  a group  of transformations of  a space  $V$  and denote by  $ \mathcal{O} = Kx$   an orbit of one point $x$ of  $V$.
  The following  obvious $K$-invariance principle holds. \\[10pt]
 {\bf Principle of  invariance.} {\it   If  $K$  is  a group  of transformations of  a space  $V$, 
      the  information,  which  a system of uniformly  distributed observers  on   an orbit $\mathcal{O} = K x \subset V$  
  sends to a
  common  center,
   is  invariant  with respect to  the transformations of  the  group   $K$.}\\[10pt] 
This principle can be used to explain the main difference between the firing  of simple and complex neurons  \cite{H,H-H,Car}.
As we recalled in \S \ref{sec:1}, n. 7,
the firing of a complex cell is
 invariant with respect to the shifts of stimuli inside of its receptive field, in contrast with
a simple cell. This can be  explained by assuming that  the RP of all the  simple cells connected to the complex cell are  derived from a  {\it  mother} RP
via transformations of    the  group $K = \bR^2$ of  translations. By the above Principle of Invariance, the
information, which is obtained by  the complex cell,  must be shift invariant.
We recall that, according to  the  reduced   model,  the simple neurons of a  small neighbourhood of   a  pinwheel   are assumed to work as 
 Gabor   filters,  which are obtained  from  the mother Gabor filter by   transformations of the  similarity
group  $\Sim(\bR^2) = \CO(2) \cdot \bR^2$. This  group contains the normal subgroup  $\bR^2$ of translation. By the above 
Principle of Invariance, the information, which  a  complex  cell  receives   from such a  system of  simple cells,  will be
translation invariant.
\par
\bigskip

\section*{Declarations}
D.A. was  supported by the Grant  Basis-Foundation Leader n  22-7-1-34-1. Besides this, no other funds, grants, or
support was received.
\par

 \end{document}